\documentclass[letterpaper,11pt]{article}

\usepackage[utf8]{inputenc}
\usepackage[letterpaper,hmargin=1in,vmargin=1in,footskip=36pt]{geometry}
\usepackage{cite} 

\usepackage{amsmath,amsthm}
\usepackage{amsfonts,amssymb}

\usepackage{paralist}

\newtheorem{theorem}{Theorem}
\newtheorem{lemma}{Lemma}

\newtheorem{definition}{Definition}
\newtheorem{observation}{Observation}

\usepackage[font=small,labelfont=bf]{caption}

\usepackage{graphicx}
\usepackage{tikz}
\usetikzlibrary{chains,fit,shapes,decorations.pathmorphing,arrows}
\usepackage{xcolor,color}

\usepackage{multirow}
\usepackage{tabularx,booktabs}
\newcolumntype{Y}{>{\raggedright\arraybackslash}X}
\newcolumntype{Z}{>{\centering\arraybackslash}X}

\usepackage{algorithmic}
\usepackage{algorithm}

\usepackage{xspace}

\newcommand{\EDF}{\ensuremath{\mathsf{EDF}}\xspace}
\newcommand{\LLF}{\ensuremath{\mathsf{LLF}}\xspace}
\newcommand{\OPT}{\ensuremath{\mathsf{OPT}}\xspace}
\newcommand{\Earlyfit}{\ensuremath{\mathsf{EarlyFit}}\xspace}
\newcommand{\Mediumfit}{\ensuremath{\mathsf{MediumFit}}\xspace}
\newcommand{\A}{\ensuremath{\mathsf{A}}\xspace}
\newcommand{\double}{{\sf Double}\xspace}

\newcommand{\OO}{\ensuremath{\mathcal{O}}}

\newcommand{\II}{\ensuremath{\mathcal{I}}}

\newcommand{\N}{\ensuremath{\mathbb{N}}}

\newcommand{\eps}{\varepsilon}

\DeclareMathOperator{\ex}{ex}

\usepackage[textsize=tiny,textwidth=2cm,shadow,color=green]{todonotes}

\title{
  New Results on Online Resource Minimization\thanks{This research was supported by the
    German Science Foundation (DFG) under contract  ME 3825/1. The
    third author was supported by the DFG within the research
    training group `Methods for Discrete Structures' (GRK 1408).}}

\author{
Lin Chen\thanks{Technische Universit\"at M\"unchen, Zentrum
   f\"ur Mathematik, Garching, Germany. Email:
   \texttt{nmegow@ma.tum.de}, \texttt{chenlin198662@gmail.com}.}
 \and Nicole Megow\footnotemark[2]  \and Kevin Schewior\thanks{Technische Universit\"at Berlin, Institut f\"ur
  Mathematik, Berlin, Germany. Email:
  \texttt{schewior@math.tu-berlin.de}.}}

\date{\today}

\begin{document}

\maketitle
\thispagestyle{empty}

\begin{abstract}
We consider the online resource minimization problem in which jobs
with hard deadlines arrive online over time at their release
dates. The task is to determine a feasible schedule on a minimum
number of machines. 

We rigorously study this problem and derive various algorithms with small constant competitive ratios for interesting restricted problem variants. 
As the most important special case, we consider scheduling jobs with
agreeable deadlines. We provide the first constant-ratio competitive
algorithm for the non-preemptive setting, which is of particular
interest with regard to the known strong lower bound of~$n$ for the
general problem. For the preemptive setting, we show that the natural
algorithm \LLF achieves a constant ratio for agreeable jobs, while in general it has a lower bound of $\Omega(n^{1/3})$.

We also give an $\OO(\log n)$-competitive algorithm for the general preemptive problem, which improves upon a known $\OO(\log (p_{max}/p_{min}))$-competitive algorithm. Our algorithm maintains a dynamic partition of the job set into loose and tight jobs and schedules each (temporal) subset individually on separate sets of machines. The key is a characterization of how the decrease in the relative laxity of jobs influences the optimum number of machines. To achieve this we derive  a compact expression of the optimum value, which might be of independent interest. We complement the general algorithmic result by showing lower bounds that rule out that other known algorithms may yield a similar performance~guarantee.

\end{abstract}

\clearpage
\setcounter{page}{1}

\section{Introduction} Minimizing the resource usage is a key  to the
achievement of economic, environmental or societal goals.
%
%
We consider the fundamental problem of minimizing the number of
machines that is necessary for feasibly scheduling jobs with
release dates and hard deadlines. We consider the online variant of this problem in which every job
becomes known to the online algorithm only at its release date. We
denote this problem as {\em online machine minimization problem}. We also
consider a {\em semi-online} variant in which we give the
online algorithm slightly more information by giving the optimal
number of machines as part of the input.

We develop several algorithms for preemptive and non-preemptive problem variants, and evaluate
their performance by  the competitive ratio, a widely-used measure that
compares the solution value of an online algorithm with an optimal
offline value. We derive the first constant-competitive 
algorithms for certain variants of (semi-)online machine minimization
and improve several others. As a major special case, we consider online deadline
scheduling with {\em agreeable deadlines}, i.e., when the order of deadlines for all jobs coincides with the order of
their release dates.  We give the first constant-competitive algorithm for scheduling agreeable jobs non-preemptively,
which contrasts to the strong lower bound of $n$ for the general
problem~\cite{Saha13}. We also prove that for the preemptive scheduling problem, one of the most natural algorithms called {\em Least Laxity First} (\LLF) achieves a constant ratio for agreeable jobs, in contrast with its lower bound of $\Omega(n^{1/3})$ on the general problem.
Our most general result is a $\OO(\log
n)$-competitive algorithm for the preemptive scheduling problem.

\paragraph{Previous results.} The preemptive semi-online machine minimization problem, in which the
optimal number of machines is known in advance, has been
investigated extensively by Phillips et al.~\cite{phillipsSTW02}, and there have
hardly been any improvements since then. Phillips et al.\ show a
general lower bound of $\frac{5}{4}$ and leave a huge gap to the
upper bound $\OO(\log \frac{p_{\max}}{p_{\min}})$ on the competitive
ratio for the so-called {\em Least Laxity First} 
(\LLF) Algorithm. Not so surprisingly, they also rule out that the {\em Earliest Deadline
  First} (\EDF) Algorithm could improve on the performance of \LLF;
indeed they show a lower bound of
$\Omega(\frac{p_{\max}}{p_{\min}})$. It is a wide open question if
preemptive (semi-)online machine minimization admits a constant-factor
online algorithm.

The non-preemptive problem is considerably harder than the preemptive
problem. If the set of jobs arrives online over time, then no algorithm
can achieve a constant or even sublinear competitive ratio~\cite{Saha13}. 
However, relevant special cases admit
online algorithms with small constant worst-case guarantees. The
problem with unit processing times was studied in a series of
papers~\cite{KleywegtNST99,ShiY08,KaoCRW12,Saha13,DevanurMPY14} and implicitly in
the context of energy-minimization in~\cite{BansalKP07}. It has been shown that an optimal online
algorithm has the exact competitive ratio~$e\approx
2.72$~\cite{BansalKP07,DevanurMPY14}. For non-preemptive scheduling of
jobs with equal deadlines, an upper bound of~$16$ is given
in~\cite{DevanurMPY14}. We are not aware of any previous work on
online machine minimization restricted to instances with agreeable
deadlines. However, in other contexts, e.g., online buffer management~\cite{JezLSS12} and scheduling
with power management~\cite{AlbersMS14,AngelBC14}, it has been studied
as an important and relevant class of instances.


Online deadline scheduling with extra speedup is another related
problem variant. We are given the same input as in the semi-online
machine minimization problem. Instead of equipping an online algorithm
with additional machines, the given~$m$ machines may run at an increased
speed~$s\geq 1$. The goal is to find an online scheduling algorithm that requires
minimum speed~$s$.  This problem seems much better understood and
(nearly) tight speedup factors below~$2$ are
known~(see~\cite{phillipsSTW02,lamT99,anandGM11}). However, the power
of speed is much stronger than additional machines since it allows
parallel processing of jobs to some extent. Algorithms that perform
well for the speed-problem, such as, \EDF, \LLF and a deadline-ordered
algorithm, do not admit constant performance guarantees for the
machine minimization problem.

We also mention that the offline problem, in
which all jobs are known in advance, can be solved
optimally in polynomial time if job preemption is allowed~\cite{horn74}. 
Again, the problem complexity increases drastically if
preemption is not allowed. In fact, the problem of  deciding whether
one machine suffices to schedule all the jobs non-preemptively is strongly
NP-complete~\cite{GareyJ77}. It is even open if a constant-factor
approximation exists; a lower bound of~$2-\eps$ was given
in~\cite{CieliebakEHWW04}. The currently best known non-preemptive approximation algorithm
is by Chuzhoy et al.~\cite{ChuzhoyGKN04} who propose a sophisticated
rounding procedure for a linear
programming relaxation that yields an approximation factor of $\OO(\sqrt{\frac{\log n}{\log\log
    n}})$. Small constant factors were obtained for special
cases~\cite{CieliebakEHWW04,YuZ09}: in particular, Yu and Zhang~\cite{YuZ09}
give a~$2$-approximation when all release dates are equal and
a~$6$-approximation when all processing times are equal.

\paragraph{Our contribution.} We develop several algorithms and
bounding techniques for the (semi-) online machine minimization
problem. We give an improved algorithm for the general problem and present a rigorous study of interesting restricted
variants for which we show small constant competitive ratios. A summary of our new results can be found in Table~\ref{tab:results}.

\begin{table}
  \begin{small}
  \definecolor{darkred}{rgb}{0.7,0,0}
  \def\arraystretch{1.15}
  \begin{tabularx}{\linewidth}{@{}cZZZZZZZZ@{}}
    \toprule
    & \multicolumn{4}{ c }{preemptive} & \multicolumn{4}{ c }{non-preemptive} \\
    & \multicolumn{2}{ c }{{\footnotesize \OPT} known} & \multicolumn{2}{ c }{{\footnotesize \OPT}
      unknown} & \multicolumn{2}{ c }{{\footnotesize \OPT} known} & \multicolumn{2}{ c
    }{{\footnotesize \OPT} unknown}\\ 
   & LB & UB & LB & UB & LB & UB & LB & UB\\ 
    \midrule
    general & $\frac{5}{4}$~\cite{phillipsSTW02}
    &\textcolor{darkred}{$\mathcal{O}\left(\log n\right)$} &
    $e$~\cite{BansalKP07,DevanurMPY14}
    &\textcolor{darkred}{$\mathcal{O}\left(\log n\right)^\star$} &
    $n$~\cite{Saha13}  & $-$ & $n$~\cite{Saha13}  & $-$ \\
   
    $p_j \leq \alpha(d_j-r_j)$ & $-$ & \textcolor{darkred}{$\frac{1}{\left(1-\alpha\right)^2}$} & $-$ &
    \textcolor{darkred}{$\frac{4}{\left(1-\alpha\right)^2}^\star$}
    & $n$~\cite{Saha13} & $-$ & $n$~\cite{Saha13} & $-$\\

    agreeable deadlines & $-$ & \textcolor{darkred}{$18$} &
    $e$~\cite{BansalKP07,DevanurMPY14} &\textcolor{darkred}{$72^\star$} & $-$
    & \textcolor{darkred}{$ 9$} & $e$~\cite{BansalKP07,DevanurMPY14} & \textcolor{darkred}{$16$} \\

    $d_j\equiv d$ & $-$ &\textcolor{darkred}{$1$} &
    $e$~\cite{BansalKP07,DevanurMPY14}& \textcolor{darkred}{$4^\star$} & $-$
    & \textcolor{darkred}{$5\frac{1}{4}$} & $e$~\cite{BansalKP07,DevanurMPY14}&\textcolor{darkred}{$11\frac{1}{9}$}\\
  
    $p_j\equiv p$ & \textcolor{darkred}{$\frac{8}{7}$} &
    \textcolor{darkred}{$3$} & $e$~\cite{BansalKP07,DevanurMPY14} &
    \textcolor{darkred}{$9.38$}  & \textcolor{darkred}{$\frac{8}{7}$} & \textcolor{darkred}{$4$}
    &$e$~\cite{BansalKP07,DevanurMPY14}& \textcolor{darkred}{$10$}\\
  
    $p_j\equiv 1$ & $-$ & $1$~\cite{KaoCRW12} & $e$~\cite{BansalKP07,DevanurMPY14} &
    $e$~\cite{BansalKP07,DevanurMPY14} & $-$ & $1$~\cite{KaoCRW12} &
    $e$~\cite{BansalKP07,DevanurMPY14} & $e$~\cite{BansalKP07,DevanurMPY14} \\
  \bottomrule
\end{tabularx}
\caption{State of the art for (semi-)online machine
  minimization. New results presented in this paper are colored red
  (and  have no reference). Results for the online setting that are
  directly derived from the semi-online setting are marked with a
  star~($^\star$). Trivial bounds are excluded ($-$).} 
  \label{tab:results}
  \end{small}
\end{table}

As an important special case, we consider the class of instances
with  {\em agreeable deadlines} in which the order of deadlines for all jobs coincides with the order of
their release dates. We give the first constant competitive
ratio for the preemptive and non-preemptive problem variants.
The constant performance guarantee for the non-preemptive
setting is particularly interesting with regard to the strong lower bound of
$n$ for instances without this restriction~\cite{Saha13}. For the preemptive setting, we show that the natural algorithm \LLF admits a constant performance guarantee on agreeable jobs, in contrast with its non-constant ratio on the general problem. We also quantify how the tightness of jobs influences the performance
of algorithms such as \EDF and \LLF, which might be of independent interest.


Our most general result is a $\OO(\log n)$-competitive algorithm for the unrestricted preemptive online
problem. This improves on the $\OO(\log
(p_{\max}/p_{\min}))$-competitive algorithm by Phillips et
al.~\cite{phillipsSTW02,phillipsSTW97}, which was actually
obtained for the semi-online problem in which the optimal
number of machines is known in advance. But as we shall see it can be
used to solve the online problem as well. In fact, we show that the
online machine minimization problem can be reduced to the semi-online
problem by losing at most a factor~$4$ in the competitive ratio. This
is true for preemptive as well as non-preemptive~scheduling.

Our algorithm for the general problem maintains a dynamic partition of jobs into
loose and tight jobs and schedules each (temporal) subset
individually on separate machines. The loose jobs are scheduled
by \EDF whereas the tight jobs require a more sophisticated
treatment. The key is a characterization of how the decrease in the
relative laxity of jobs influences the optimum number of machines. To
quantify this we derive a compact expression of the optimum value,
which might be of independent interest. 
We complement this algorithmic result by showing strong lower bounds on \LLF or any deadline-ordered algorithm. 

We remark that our algorithms run in polynomial time (except for the non-preemptive online algorithms). As a side
result, we even improve upon a previous approximation result for
non-preemptive offline machine minimization with jobs of equal processing time when the optimum is not
known. Our semi-online $4$-competitive algorithm can be easily turned into an offline~$4$-approximation which
improves upon a known~$6$-approximation~\cite{YuZ09}.

\paragraph{Outline.} In Section~\ref{mp-sec:prelim}, we introduce some notations and state basic results used throughout the paper. In Section~\ref{sec:special-cases1}, we give constant-competitive algorithms for the problem of scheduling agreeable jobs. We continue with the other special cases of equal processing times and uniform deadlines in Sections~\ref{sec:special-cases2} and~\ref{sec:special-cases3}, respectively. Our general $\OO(\log n)$-competitive
algorithm is presented in Section~\ref{sec:general logn}. We finally present some lower bounds in Section~\ref{sec:LB}.

\section{Problem Definition and Preliminaries}
\label{mp-sec:prelim}

\paragraph{Problem definition.} Given is a set
of jobs~$J=\{1,2,\ldots,n\}$ where each job~$j \in J$ has a processing
time~$p_j\in \N$, a release
date~$r_j\in \N$ which is the earliest possible time at which the job can be
processed, and a deadline~$d_j\in \N$ by which it must be
completed. The task is to open a minimum number of machines such that
there is a feasible schedule in which no job misses its deadline. In a
feasible schedule each job~$j\in J$ is scheduled for~$p_j$ units of
time within the time window~$[r_j,d_j]$. Each opened machine can process at most one job at
the time, and no job is running on multiple machines at the same
time. When we allow job preemption, then a job can be preempted at any
moment in time and may resume processing later on the same or any
other machine. When preemption is not allowed, then a job must run
until completion once it has started. 

To evaluate the performance of our online algorithms we perform a {\em
  competitive analysis} (see e.g.~\cite{borodinEY98}). We call an
online algorithm~\A  $c$-{\em competitive} if~$m_{\A}$ machines with~$m_{\A} \leq c\cdot m$ suffice to
guarantee a feasible solution for any instance that admits a
feasible schedule on~$m$ machines.

\paragraph{Notation.} 
We let~$J(t):=\{j\in J\,|\, r_j\leq t\}$ denote the set of all jobs that have been released by time~$t$. For some job set~$J$ we let $m(J)$ denote the minimum number of machines needed to feasibly schedule~$J$. If $J$ is clear from the context, we also write more concisely~$m$ for $m(J)$ and $m(t)$ for $m(J(t))$. 

Consider an algorithm $\mathsf{A}$ and the schedule~$\mathsf{A}(J)$ it obtains for $J$. We denote the {\em remaining processing time} of a job $j$ at time $t$ by $p_j^{\mathsf{A}(J)}(t)$. 
Further, we define the {\em laxity} of~$j$ at time~$t$ as~$\ell^{\mathsf{A}(J)}_j(t)=d_j-t-p_j^{\mathsf{A}(J)}(t)$ and call $\ell_j=\ell_j(r_j)$ the {\em original laxity}. We denote the largest deadline in $J$ by $d_{\max}(J)=\max_{j\in J}{d_j}$. 
We classify jobs by the fraction that their processing time takes from the feasible time window and call a job $j$ $\alpha$-{\em loose} at $t$ if $p_j(t)\le \alpha (d_j-r_j)$ or $\alpha$-{\em tight} at $t$, otherwise. Here, we drop $t$ if $t=r_j$. 
Finally, a job is called {\em active} at any time it is released but unfinished.

We analyze our algorithms by estimating the total processing volume (or workload) that is assigned to certain time intervals~$[t,t^\prime)$. To this end, we denote by $w_{\mathsf{A}(J)}(t)$ the total processing volume of $J$ that~$\mathsf{A}$ assigns to $[t,t+1)$ (or simply to $t$). Since our algorithms make decisions at integral time points, this is simply the number of assigned jobs. Further, $W_{\mathsf{A}(J)}(t)$ denotes the total remaining processing time of all unfinished jobs (including those not yet released). 
We omit the superscripts whenever $J$ or $\mathsf{A}$ are non-ambiguous, and we use $\OPT$ to indicate an optimal~schedule.

\paragraph{Reduction to the semi-online problem.} In the following we show that we may assume that the optimum number of
machines~$m$ is given by losing at most a factor~$4$ 
in the competitive ratio. 

We employ the general idea of {\em doubling} an unknown parameter,
which has been successfully employed in solving various online optimization
problems~\cite{chrobakK06}. In our case, the unknown parameter is
the optimal number of machines. The idea is to open additional
machines at any time that the optimum solution has doubled.

Let $\mathsf{A}_{\alpha}(m)$ denote an~$\alpha$-competitive algorithm for the
semi-online machine minimization problem given the optimum number of
machines~$m$. Then our algorithm for the online problem is as follows.

\medskip
{\bf Algorithm \double}: 
\begin{compactitem}
\item Let~$t_0=\min_{j\in   J}r_j$. 
  For $i=1,2,\ldots$ let $t_i=\min\{t\,|\, m(t) > 2m(t_{i-1})\}$. 
\item At any time~$t_i$, $i=0,1,\ldots$, open $2\alpha
  m(t_i)$ additional machines. All jobs with $r_j\in [t_{i-1},t_i)$
  are scheduled by Algorithm $\mathsf{A}_{\alpha}(2m(t_{i-1}))$  on the
  machines opened at time~$t_{i-1}$.
\end{compactitem}
\medskip

Observe that this procedure can be executed online, since the time
points~$t_0,t_1,\ldots$ as well as $m(t_0),m(t_1),\ldots$ can be computed online and
$\mathsf{A}_{\alpha}$ is assumed to be an algorithm for the semi-online problem. Since
the optimal solution increases only at release dates, which are
integral, we may assume that~$t_i\in \N$. Notice also that \double does
not preempt jobs which would not have been preempted by Algorithm~$\mathsf{A}_{\alpha}$.

\begin{theorem}\label{thm: black box}
  Given an $\alpha$-competitive algorithm for the (non)-preemptive semi-online machine minimization problem, \double is $4\alpha$-competitive for  (non)-preemptive online machine minimization. \end{theorem}
\begin{proof}
  Let~$t_0,t_1,\ldots,t_k$ denotes the times at which \double opens new machines.
  \double schedules the jobs $J_{i} = \{j\,|\, r_j\in
  [t_{i},t_{i+1})\}$, with $i=0,1,\ldots,k$, using
  Algorithm~$A_{\alpha}(2m(t_i))$ on $2\alpha m(t_{i})$
  machines exclusively. This yields a feasible schedule since an
  optimal solution for $J_{i}\subseteq J(t_{i+1}-1)$ requires at
  most~$m(t_{i+1}-1)\leq 2m(t_i)$ machines and the
  $\alpha$-competitive subroutine~$A_{\alpha}(2m(t_i))$ is
  guaranteed to find a feasible solution given a factor $\alpha$ times
  more machines than optimal, i.e.,  $2\alpha m(t_i)$.

  It remains to compare the number of machines opened by \double with
  the optimal number of machines~$m$, which is at least~$m(t_k)$. By
  construction it holds that $2m(t_i) \leq m(t_{i+1})$, 
  which implies by recursive application that 
  \begin{equation}\label{eq:a}
    2m(t_i) \leq \frac{m(t_k)}{2^{k-i-1}}.
  \end{equation}
  The total number of machines opened by \double is
  \begin{align*}
    \sum_{i=0}^k 2\alpha m(t_i) \ \leq \  \sum_{i=0}^k \alpha
    \frac{1}{2^{k-i-1}} m(t_k) \ 
    = \  2\alpha \sum_{i=0}^k \
    \frac{1}{2^{i}} m(t_k)
    \ \leq \ 4 \alpha m,
   \end{align*}
   which concludes the proof.
   \end{proof}

Using standard arguments, this performance guarantee can be improved to a factor~$e\approx 2.72$ using randomization over the doubling factor. However, in the context of hard real-time guarantees a worst-case guarantee that is achieved in expectation seems less relevant and we omit it.

\paragraph{Algorithms.} An algorithm for the (semi-)online machine minimization problem is called {\em busy} if, at all times $t$, $w_{\mathsf{A}}(t)<m$ implies that there are exactly $w_{\mathsf{A}}(t)$ active jobs at $t$. Two well-known busy algorithms are the following. The {\em Earliest Deadline First} ($\EDF$) algorithm on~$m'$ machines schedules at any time the~$m'$ jobs with the smallest deadline. 
The {\em  Least Laxity First} ($\LLF$) algorithm on~$m'$ machines, schedules at any time the~$m'$ jobs with the smallest non-negative laxity at that time. In both cases, ties are broken in favor of earlier release dates or, in case of equal release dates, by lower job indices. For simplicity, we assume in this paper that both \EDF and \LLF make decisions only at integer time points.

\paragraph{Scheduling loose jobs.} \label{mp-subsec: oa-edf-good}
Throughout the paper, we will use the fact that, for any fixed
$\alpha<1$, we can achieve a constant competitive ratio on
$\alpha$-loose jobs, i.e., when $p_j(t)\le \alpha (d_j-r_j)$, simply by using $\mathsf{EDF}$. We remark that the same result also applies to \LLF. 

\begin{theorem}\label{thm: EDF-small}
If every job is $\alpha$-loose, \EDF is a $1/(1-\alpha)^2$-competitive algorithm for preemptive semi-online machine minimization.
\end{theorem}
We prove this theorem by establishing a contradiction based on the workload that \EDF when missing deadlines must assign to some interval.
We do so by using the following work load inequality, which holds for arbitrary busy algorithms.

\begin{lemma}\label{lemma: small-busy-load}
Let every job be $\alpha$-loose, $c\ge 1/(1-\alpha)^2$ and let $\mathsf{A}$ be a busy algorithm for semi-online machine minimization using $cm$ machines. Assume that $\ell_j^\mathsf{A}(t^\prime)\geq 0$ holds, for all $t^\prime\leq t$ and $j$. Then $$W_\mathsf{A}(t)\le W_{\OPT}(t)+\frac{\alpha}{1-\alpha}\cdot m\cdot\left(d_{\max}-t\right).$$
\end{lemma} 

\begin{proof}
We prove by induction. The base is clear since $W_A(0)=W_{\OPT}(0)$. Now assume the lemma holds for all $t^\prime<t$. There are three possibilities.

\noindent\textbf{Case 1.} We have $w_\mathsf{A}(t)\le \alpha/(1-\alpha)\cdot m$. By the fact that $\mathsf{A}$ is busy, there is an empty machine at time $t$, meaning that there are at most $\alpha/(1-\alpha)\cdot m$ active jobs. Given that $\ell_j^\mathsf{A}(t^\prime)\geq 0$ for all $t^\prime\leq t$, each of the active jobs has a remaining processing time of no more than $d_{\max}-t$, implying that $W_\mathsf{A}(t)$ is bounded by $\alpha/(1-\alpha)\cdot m\cdot(d_{\max}-t)$ plus the total processing time of the jobs that are not released. As the latter volume cannot exceed $W_{\OPT}(t)$, the claim follows.

\noindent\textbf{Case 2.} We have $w_\mathsf{A}(t)\ge  1/(1-\alpha)\cdot m$. According to the induction hypothesis, it holds that $W_\mathsf{A}(t-1)\le W_{\OPT}(t-1)+\alpha/(1-\alpha)\cdot m\cdot (d_{\max}-t+1)$. Using $w_{\OPT}(t)\le m$, we get $W_{\OPT}(t)\ge W_{\OPT}(t-1)-m.$ For Algorithm $\mathsf{A}$, it holds that $W_\mathsf{A}(t)=W_\mathsf{A}(t-1)-w_\mathsf{A}(t)\le W_\mathsf{A}(t-1)-1/(1-\alpha)\cdot m.$ Inserting the first inequality into the second one proves the claim.

\noindent\textbf{Case 3.} We have $\alpha/(1-\alpha)\cdot m<w_\mathsf{A}(t)<1/(1-\alpha)\cdot m$. Again by the fact that $\mathsf{A}$ is busy, there is an empty machine at time $t$, i.e., there are less than $1/(1-\alpha)\cdot m$ active jobs. Distinguish two cases.

\noindent\textbf{Case 3a.} We have $p_j(t)\le \alpha (d_j-t)$ for all active jobs. Then the total remaining processing time of active jobs is bounded by $\alpha (d_{\max}-t)\cdot 1/(1-\alpha)\cdot m$. Plugging in the total processing time of unreleased jobs, which is bounded by $W_{\OPT}(t)$, we again get the claim.

\noindent\textbf{Case 3b.} There exists an active job $j$ with $p_j(t)> \alpha (d_j-t)$. Using that $j$ is $\alpha$-loose, i.e., $p_j\le \alpha(d_j-r_j)$, we get that $j$ is not processed for at least $(1-\alpha)(t-r_j)$ many time units in the interval $[r_j, t]$. As Algorithm $\mathsf{A}$ is busy, this means that all machines are occupied at these times, yielding
\begin{align*}
W_\mathsf{A}(t)&\le W_\mathsf{A}(r_j)-(1-\alpha)(t-r_j)\cdot cm\\
&\le W_\OPT(r_j)+\frac{\alpha}{1-\alpha}\cdot m\cdot(d_{\max}-r_j)-(1-\alpha)(t-r_j)\cdot cm\\
&\le W_\OPT(r_j)+\frac{\alpha}{1-\alpha}\cdot m\cdot(d_{\max}-t) - m\cdot(t-r_j),
\end{align*}
where the second inequality follows by the induction hypothesis for $t=r_j$, and the third one follows by $c\ge 1/(1-\alpha)^2$. Lastly, the feasibility of the optimal schedule implies $$W_{\OPT}(t)\ge W_{\OPT}(r_j)-m\cdot (t-r_j),$$ which in turn implies the claim by plugging it into the former inequality.
\end{proof}

\begin{proof}[Proof of Theorem~\ref{thm: EDF-small}.]
Let $cm$ be the number of machines \EDF uses and consider the schedule produced by \EDF from an arbitrary instance $J$ only consisting of $\alpha$-loose jobs. We have to prove that every job is finished before its deadline if $c=1/(1-\alpha)^2$. To this end, suppose that \EDF fails. Among the jobs that are missed, let $j^\star$ be one of those with the earliest release date.

First observe that we can transform $J$ into $J^\prime\subset J$ such that $j^\star\in J^\prime$, $\max_{j\in J^\prime} d_j = d_{j^\star}$ and \EDF still fails on $J^\prime$. For this purpose, simply define $J^\prime=\left\{j\in J\mid d_j\leq d_{j^\star}\right\}$ and notice that we have only removed jobs with a later deadline than $\max_{j\in J^\prime} d_j$. This implies that every job from $J^\prime$ receives the same processing in the \EDF schedule of $J^\prime$ as in the one of $J$.

We can hence consider $J^\prime$ instead of $J$ from now on. In the following, we establish a contradiction on the workload during the time interval $[r_{j^\star},d_{j^\star}]$. The first step towards this is applying Lemma~\ref{lemma: small-busy-load} for an upper bound. Also making use of the feasibility of the optimal schedule, we get:
\begin{align*}
W_{\EDF}(r_{j^\star})&\le W_{\OPT}(r_a)+\frac{\alpha}{1-\alpha}\cdot m\cdot(d_{\max}-r_{j^\star})\\
&\le \frac{1}{1-\alpha}\cdot m\cdot(d_{\max}-r_{j^\star}).
\end{align*}
We can, however, also lower bound this workload. Thereto, note that job $j^\star$ is not processed for at least $(1-\alpha)(d_{j^\star}-r_{j^\star})$ time units, implying that all machines must be occupied by then, i.e., \[W_{\EDF}(r_{j^\star})> (1-\alpha)\cdot cm \cdot (d_{\max}-r_{j^\star}).\] If we compare the right-hand sides of the two inequalities, we arrive at a contradiction if and only if $c\geq 1/(1-\alpha)^2$, which concludes the proof.
\end{proof}

\section{Constant-competitive Algorithms for Agreeable Deadlines}
\label{sec:special-cases1}

In this section, we study the special case of (semi-)online machine minimization of jobs with {\em agreeable deadlines}. That is, for any two jobs $j,k\in J$ with $r_j< r_k$ it holds that also $d_j\leq d_k$. 

\subsection{Preemptive scheduling}

We propose an algorithm that again schedules~$\alpha$-loose jobs by \EDF or \LLF on $m/(1-\alpha)^2$ machines, which is feasible by Theorem~\ref{thm: EDF-small}. Furthermore, $\alpha$-tight jobs are scheduled by \LLF, for which we obtain the following bound.

\begin{theorem}\label{thm1:LLF-agreeable}
  The competitive ratio of \LLF is at most $4/\alpha+6$ if all jobs are $\alpha$-tight and have agreeable deadlines.
\end{theorem}

 The proof is based on a load argument. Notice that
 Lemma~\ref{lemma: small-busy-load} is no longer true since jobs
 now are all tight. To define a new load inequality we introduce the
 set $\Lambda(t)$ which is defined as follows. Consider the latest
 point in time before (or at) $t$ such that at most $m$ jobs are
 processed and denote it by $\phi(t)=\max\{t'\le t\mid w_{\LLF}(t')\le
 m\}$. Let $\Lambda({t})$ be the set of active jobs at time
 $\phi(t)$. Obviously $|\Lambda(t)|\le m$. If there does not exist
 such a $\phi(t)$, i.e., $w(x)>m$ holds for any $x\le t$, then we let
 $\Lambda(t)=\emptyset$. 
\begin{lemma}
\label{lemma:mp-agreeable-LLF-load}
Consider an instance consisting only of $\alpha$-tight jobs. Let $c\ge 4/\alpha+6$ and $t$ such that there is no job missed by \LLF up to time $t$. Then we have $$W_{\LLF}(t)\le W_{\OPT}(t)+\sum_{j\in \Lambda(t)}p_j(t).$$
\end{lemma}


\begin{proof}
The lemma is trivially true for $t=0$. Now assume $t>0$ and that the lemma holds for all $t'<t$. We make the following assumptions. Indeed, if any of them is not true, then we can easily prove the lemma.

\begin{enumerate}[(i)]
\item $\phi(t)=t_l$ exists, and $t_l<t$.
\item $\Lambda(t)\neq \emptyset$.
\item During $[t_l,t]$, at least one job is not processed all the time.
\item $d_j\ge t$ for at least one job $j\in\Lambda(t)$.
\end{enumerate}

Consider assumption (i). The lemma is obviously true if $\phi(t)$ does not exist because that means $w_{\LLF}(t')> m$ for any $t'\le t$. That is at time $t$, \LLF must have finished more load than the optimum solution, in which only $m$ machines are available. It follows directly that $W_{\LLF}(t)\le W_{\OPT}(t)$.
Otherwise there exists $\phi(t)=t_l\le t$. Notice that, if $t_l=t$, then obviously $W_{\LLF}(t_l)\le W_{\OPT}(t_l)+\sum_{j\in \Lambda(t_l)}p_j(t_l)$. 
Thus we assume in the following that $t_l<t$, and it follows that $\Lambda(t')=\Lambda(t_l)$ for any $t'\in[t_l,t]$. 

Consider assumption (ii). If $\phi(t)=t_l$ exists while $\Lambda(t)$ is empty, then \LLF finishes all the jobs released before time $t_l$, while $w_{\LLF}(t')>m\ge w_{\OPT}(t')$ for $t_l< t'\le t$. It is easy to see that the lemma follows. 

Consider assumption (iii). Notice that $W_{\LLF}(t)\le W_{\OPT}(t)+\sum_{j\in\Lambda(t)} p_j(t)$ is obviously true if no new job is released after $t_l$. If a new job is released after $t_l$ and scheduled without preemption by \LLF, then its contribution to $W_{\LLF}(t)$ is no more than its contribution to $W_{\OPT}(t)$. Plugging in the contributions of the newly released jobs, we still have $W_{\LLF}(t)\le W_{\OPT}(t)+\sum_{j\in\Lambda(t)} p_j(t)$.

Consider assumption (iv). If it is not true, we have $W_{\LLF}(t-1)\le W_{\OPT}(t-1)$ according to the induction hypothesis, and $p_j(t-1)=0$ holds for all $j \in \Lambda(t-1)$. Given that $W_{\LLF}(t)\le W_{\LLF}(t-1)-m$ and $W_{\OPT}(t-1)\ge W_{\OPT}(t)-m$, the lemma is true. 

We proceed with the four assumptions above. 
Notice that some job is not processed in $[t_l,t]$, and we let $[t_0,t_0+1]$ be the first such time. We let $t_a$ be the first time in $[t_0,t]$ with $w_{\LLF}(t_a)\le 2m$. Notice that $t_a$ might not exist, and in this case $w_{\LLF}(t')>2m$ holds for all $t'\in [t_0,t]$. As a consequence we get $W_{\LLF}(t)\le W_{\LLF}(t_0)-2m\cdot(t-t_0).$ On the other hand, $W_{\OPT}(t)\ge W_{\OPT}(t_0)-m\cdot(t-t_0).$ Given the fact that $\sum_{j\in\Lambda(t)} p_j(t)\le \sum_{j\in\Lambda(t)} p_j(t_0) -m\cdot(t-t_0),$ we have $W_{\LLF}(t)\le W_{\OPT}(t)+\sum_{j\in\Lambda(t)} p_j(t)$.

Now assume that $t_a\le t$ exists. Let $t_b\leq t_0$ be the latest time with $w_{\LLF}(t_b-1)\le \gamma cm$ for some parameter $\gamma<1$ that will be fixed later (i.e., $w(t')>\gamma cm$ for $t'\in[t_b,t_0]$). Given that $w_{\LLF}(t_l)\le m$, we know that $t_b\in [t_l,t_0]\subseteq [t_l,t]$, and furthermore that
\begin{equation}\label{eq:load-ineq 1}
W_{\LLF}(t)\le W_{\LLF}(t_b)-\gamma cm\cdot (t_0-t_b).
\end{equation}

Since $w_{\LLF}(t_b-1)\le \gamma cm$ and $w_{\LLF}(t_0)=cm$, at least $(1-\gamma)\cdot cm$ jobs are released during $[t_b,t_0]$. Since $w_{\LLF}(t_a)\le 2m$, among these $(1-\gamma)\cdot cm$ jobs, there are at least $(1-\gamma-2/c)\cdot cm$ finished before $t_a$. Let $j$ be any of these jobs released after $t_b$ and finished before $t_a$. Now consider the active jobs at time $t_l$. We know by (iv) that at least one of them, say job $j'$, has a deadline no less than $t$. Recall that we are given an agreeable instance, hence job $j$ that is released after $j'$ must have a larger deadline, i.e., $d_j\ge t$. Further, we have $r_j\le t_0$, and as job $j$ is a tight job it follows that $p_j\ge \alpha\cdot(d_j-r_j)\ge \alpha\cdot(t-t_0)$. Thus,
\begin{equation}\label{eq:load-ineq 2}
W_{\LLF}(t)\le W_{\LLF}(t_b)-(1-\gamma-2/c)\cdot\alpha cm\cdot(t-t_0).
\end{equation}

By combining Inequalities~\ref{eq:load-ineq 1} and~\ref{eq:load-ineq 2}, we have
\begin{align}\label{eq:load-ineq 3}
W_{\LLF}(t)&\le W_{\LLF}(t_b)-\max\{\gamma\cdot cm\cdot(t_0-t_b), (1-\gamma-2/c)\cdot\alpha cm\cdot(t-t_0)\}\notag\\
&\le W_{\LLF}(t_b)-1/2\cdot\gamma cm\cdot(t-t_b).
\end{align}
if we choose $\gamma$ such that $\alpha\cdot(1-\gamma-2/c)=\gamma$.

On the other hand, it is easy to observe the following three inequalities: $W_{\LLF}(t_b)\le W_{\OPT}(t_b)+\sum_{j\in\Lambda(t)} p_j(t_b)$, $W_{\OPT}(t)\ge W_{\OPT}(t_b)-m\cdot(t-t_b),$ and
$\sum_{j\in\Lambda(t)}  p_j(t)\ge \sum_{j\in\Lambda(t)} p_j(t_b)-m\cdot(t-t_b)$. Hence it follows that
\begin{equation}\label{eq:load-ineq 4}
W_{\LLF}(t_b)\le W_{\OPT}(t)+m\cdot(t-t_b)+\sum_{j\in\Lambda(t)} p_j(t)+m\cdot(t-t_b).
\end{equation} 
By combining Inequalites~\ref{eq:load-ineq 3} and~\ref{eq:load-ineq 4}, it can be easily seen that the lemma holds if we set
$\gamma/2\cdot c\ge 2$.
Combining this inequality with the equation $\alpha\cdot(1-\gamma-2/c)=\gamma$ above, we get $c\ge 4/\alpha+6$, i.e., for $c\ge 4/\alpha+6$, the lemma is true.
\end{proof}

\begin{proof}[Proof of Theorem~\ref{thm1:LLF-agreeable}] Assume the contrary, i.e., \LLF fails. From the counterexamples, we pick one with the minimum number of jobs and let $J$ denote the instance. Among all those jobs that are missed, we let $z$ be the job with the earliest release date. Since job $z$ is missed, there must exist some time $t_0$ when the laxity of job $z$ is $0$ but it is still not processed in favor of $cm$ other jobs. Let these jobs be job $1$ to job $cm$ with $r_1\le r_2\le\cdots\le r_{cm}$ and $S_z$ be the set of them. We have the following observations used throughout the proof:
\begin{itemize}
\item No job is released after $t_0$.
\item For $1\le i\le cm$, we have $d_i\le d_z$.
\item For $1\le i\le cm$, each job $i$ has $0$ laxity at $t_0$ and is not missed. Hence it is processed without preemption until time $d_i$.
\end{itemize}
  
Let $t_b=r_{\lfloor cm/3\rfloor}$ and $t_a=d_{\lfloor 2cm/3\rfloor}$. 
Consider time $t_b$. According to Lemma~\ref{lemma:mp-agreeable-LLF-load}, we have
\begin{eqnarray}\label{eq1:mp-load-LLF-agreeable}
W_{\LLF}(t_b)\le W_{\OPT}(t_b)+\sum_{j\in\Lambda(t_b)}p_j(t_b)\le W_{\OPT}(t_b)+m\cdot(t_a-t_b).
\end{eqnarray} 

To see why the second inequality holds, observe that there are $\lfloor cm/3\rfloor>m$ jobs being processed at time $t_b$. Thus any job in $\Lambda(t_b)$ must have a release date no later than $r_m$ and hence it has a deadline no later than $d_m\le t_a$. Furthermore, no job from $\Lambda(t_b)$ is missed, so the total remaining load of these jobs at $t_a$ is at most $m\cdot (t_a-t_b)$.

We focus on the interval $[t_b,t_a]$ and will establish a contradiction on the load processed by \LLF during this interval. 
We define a job $j$ to be a {\em bad} job if it contributes more to $W_{\OPT}(t_a)$ than to $W_{\LLF}(t_a)$, i.e., either one of the following is true:
\begin{itemize}
\item Job $j$ is not finished at $t_a$ in both \LLF and \OPT, and we have $p_j^{\OPT}(t_a)>p_j^{\LLF}(t_a)$.
\item At $t_a$, job $j$ is finished in \LLF but it is not finished in \OPT.
\end{itemize}

\begin{observation}
If the laxity of a job becomes $0$ in \LLF at some time during $[t_b,t_a]$, then it is not a bad job.
\end{observation}

Hence, any job in $S_z$ cannot be a bad job since their laxities are $0$ at $t_0\in[t_b,t_a]$. Since we have $t_a=d_{\lfloor 2/3cm\rfloor}$, and a bad job $j$ is not finished in OPT at time $t_a$, we get that $d_j\ge d_{\lfloor 2/3cm\rfloor}$, and $r_j\ge r_{\lfloor 2/3cm\rfloor}\ge t_b=r_{\lfloor cm/3\rfloor}$. Now delete all the bad jobs in the schedules produced by \LLF and \OPT and denote their corresponding remaining loads at some $t$ by $\hat{W}_{\LLF}(t)$ and $\hat{W}_{\OPT}(t)$, respectively. By the preceding observations, we again have $$\hat{W}_{\LLF}(t_b)\le \hat{W}_{\OPT}(t_b)+m\cdot(t_a-t_b)$$ because the same amount is subtracted from both sides of Inequality~\ref{eq1:mp-load-LLF-agreeable}. Further, by the definition of bad jobs, we have $$\hat{W}_{\LLF}(t_a)\ge \hat{W}_{\OPT}(t_a).$$

It follows directly that $$\hat{W}_{\LLF}(t_b)-\hat{W}_{\LLF}(t_a)\le 2m\cdot(t_a-t_b),$$ and we show in the following that $$\hat{W}_{\LLF}(t_b)-\hat{W}_{\LLF}(t_a)\ge \lfloor cm/3\rfloor\cdot(t_a-t_b)>2m\cdot(t_a-t_b),$$ which will yield a contradiction and thus prove the theorem. Towards proving the claim, let $\hat{w}_{\LLF}(t)$ be the load of \LLF at time $t$ in the schedule of \LLF after deleting bad jobs. We prove that for any $t\in [t_b,t_a]$, we have $\hat{w}_{\LLF}(t)\ge \lfloor cm/3\rfloor$.

We first observe that at any $t\in [t_0,t_a]$ job $\lfloor 2cm/3\rfloor+1$ to job $cm$ of $S_z$ are processed. Thus the statement is true for every $t\in [t_0,t_a]$.

Consider any time $t\in [t_b,t_0]$. The first $\lfloor cm/3\rfloor$ jobs from $S_z$ are released and not finished, which implies that $w_{\LLF}(t)\ge \lfloor cm/3\rfloor$. Obviously if at time $t$ no bad job is processed, then $\hat{w}_{\LLF}(t)=w_{\LLF}(t)$ and we are done. 
Otherwise at time $t$ some bad job is being processed. We claim that, during $[t_b,t_0]$, the jobs $1,\dots,\lfloor cm/3\rfloor$ from $S_z$ will never be preempted in favor of a bad job, and then $\hat{w}_{\LLF}(t)\ge \lfloor cm/3\rfloor$ follows. Suppose this is not true. Then there exists some time $t\in [t_b,t_0]$, some bad job $j$ and some other job $j'\in\{1,2,\cdots,\lfloor cm/3\rfloor\}$ such that $j'$ is preempted at $t$ while $j$ is being processed. We focus on the two jobs $j$ and $j'$, and pick up the last time $t'$ such that $j'$ is preempted while $j$ is being processed, and obviously $t'< t_0$. Since $j'$ is preempted at $t'$, we get $\ell_{j}(t')\le \ell_{j'}(t')$. We know that the laxity of a job decreases by $1$ if it is preempted, and remains the same if it gets processed. From time $t'+1$ to $t_0$, we know that $j'$ will never be preempted in favor of $j$ since $t'$ is the last such time. This means whenever $\ell_{j'}(t)$ decreases by $1$ (meaning that $j'$ is preempted), then $\ell_j(t)$ also decreases by $1$. So from $t'+1$ to $t_0$, $\ell_{j'}(t)$ decreases no more than $\ell_{j}(t)$, implying that $\ell_{j}(t_0)\le \ell_{j'}(t_0)=0$, which contradicts the fact that $j$ is a bad job.
\end{proof}

Recall that $\alpha$-loose jobs can be scheduled via \EDF or \LLF on $m/(1-\alpha)^2$ machines. If we set $\alpha=1/2$ and use \LLF on the $\alpha$-tight jobs, we get the following result.

\begin{theorem}
  If deadlines are agreeable, there is a $18$-competitive algorithm for preemptive semi-online machine minimization.
\end{theorem}

By applying \double, we get the following result for the fully online case.

\begin{theorem}
  If deadlines are agreeable, there is a $72$-competitive algorithm for preemptive online machine minimization.
\end{theorem}

\subsection{Non-preemptive Scheduling} Consider the non-preemptive version of machine minimization with agreeable deadlines. Again we schedule $\alpha$-tight and $\alpha$-loose jobs separately. In particular, $\alpha$-loose jobs are scheduled by a non-preemptive variant of \EDF, that is, at any time $[t,t+1]$, jobs processed at $[t-1,t]$ and are not finished continue to be processed. If there are free machines, then among all the jobs that are not processed we select the ones with earliest deadlines and process them. To schedule $\alpha$-tight jobs, we use \Mediumfit:

\medskip
{\bf Algorithm \Mediumfit}:
Upon the release of a job $j$, we schedule it non-preemptively during the medium part of its processing interval, i.e., $[r_j+1/2\cdot \ell_j, d_j-1/2\cdot \ell_j]$.
\medskip

We first prove a lemma about \EDF and the $\alpha$-tight jobs.

\begin{lemma}\label{lemma:agreeable-loose-edf}
  Non-preemptive \EDF is a $1/(1-\alpha)^2$-competitive algorithm for non-preemptive semi-online machine minimization problem if all jobs are $\alpha$-loose and have agreeable deadlines.
\end{lemma} 
\begin{proof}
We use the same technique as Theorem~\ref{thm: EDF-small}, i.e., we prove via contradiction on the workload.

We observe that the load inequality provided by Lemma~\ref{lemma: small-busy-load} is also true for non-preemptive machine minimization. Suppose \EDF fails, then among all the jobs that are missed we consider the job with the earliest release date and let it be $j^\star$. Again let $J'=\{j\in J\mid d_j\le d_{j^\star}\}$. We claim that \EDF also fails on $J'$. To see why, notice that jobs in $J\setminus J'$ have a larger deadline than $j^\star$. Hence according to the agreeable deadlines their release dates are larger or equal to $r_{j^\star}$. Furthermore according to \EDF they are of a lower priority in scheduling, thus the scheduling of job $j^\star$ could not be delayed in favor of jobs in $J\setminus J'$. This implies that \EDF also fails on $J'$. The remaining argument is the same as the proof of Theorem~\ref{thm: EDF-small}. 
\end{proof}

We remark that for non-agreeable deadlines, \EDF is not a constant-competitive algorithm for $\alpha$-loose jobs. Indeed, non-preemptive \mbox{(semi-)online} machine minimization 
does not admit algorithms using less than~$n$ machines when $m=1$,  even if every job is $\alpha$-loose for any constant $\alpha$~\cite{Saha13}. The above argument fails as we can no longer assume that the job $j^\star$ missed by \EDF has the largest deadline. 

We continue with a lemma about \Mediumfit:

\begin{lemma}\label{lemma:agreeable-medium-fit}
\Mediumfit is a $(2\lceil 1/\alpha\rceil+1)$-competitive (semi-)online algorithm for non-preemptive machine minimization on $\alpha$-tight jobs with agreeable deadlines.
\end{lemma}  
\begin{proof}
Suppose that the algorithm runs more than $\left[2\lceil 1/\alpha\rceil+1\right]\cdot m$ jobs at some time $t$. By the generalized pigeonhole principle, $2\lceil 1/\alpha\rceil+2$ of them must run on the same machine in \OPT. Consequently, on this machine, $\lceil 1/\alpha\rceil+1$ of them are not finished after or not started before $t$. W.l.o.g.~assume that the latter is the case and call these jobs $J^\star$ (the other case is symmetric).

Observe that, as $J^\star$ has been completely released at $t$, we can assume that, in \OPT, these jobs are scheduled in order of their release dates, for this is the order of their deadlines. Now choose $j^\star$ with maximum release date, that is, also maximum deadline, from $J^\star$. It holds that $$t+\sum_{j\in J^\star}p_j\leq d_{j^\star}=r_{j^\star}+\ell_{j^\star}+p_{j^\star}.$$ Plugging in $r_{j^\star}\leq t-\ell_{j^\star}/2$, which follows from the definition of \Mediumfit, yields $$\sum_{j\in J^\star\setminus\{j^\star\}}p_j\leq\frac{\ell_{j^\star}}{2}.$$ Hence, for all $j\in J^\star\setminus\{j^\star\}$, we have $$r_{j}\leq r_{j^\star}\leq t-\frac{\ell_{j^\star}	}{2}\leq t-\sum_{j'\in J^\star\setminus\{j^\star\}}p_{j'},$$ which follows by our choice of $j^\star$ as well as the previous inequalities. Since all of these jobs are $\alpha$-tight, we have $$p_j\geq\alpha\cdot(t-r_j)\geq\alpha\cdot\sum_{j'\in J^\star\setminus\{j^\star\}}p_{j'}.$$ By summing over all such $j$, we get 
\begin{align*}
\sum_{j\in J^\star\setminus\{j^\star\}}p_{j}&>\alpha\cdot(|J^\star|-1)\cdot\sum_{j\in J^\star\setminus\{j^\star\}}p_{j}\\
&\geq \alpha\cdot\left\lceil\frac{1}{\alpha}\right\rceil\cdot\sum_{j\in J^\star\setminus\{j^\star\}}p_{j}\geq \sum_{j\in J^\star\setminus\{j^\star\}}p_{j},	
\end{align*}
which is a contradiction.
\end{proof}

Combining the above two lemmas and choosing $\alpha=1/2$, we have a guarantee for the semi-online case.

\begin{theorem}
  For non-preemptive machine minimization with agreeable deadlines, there is a semi-online $9$-competitive algorithm.
\end{theorem}

A fully online $(2\cdot\lceil 1/\alpha\rceil+1+4/(1-\alpha)^2)$-competitive algorithm can be obtained by using \Mediumfit for tight jobs and combining \EDF with \double. Setting $\alpha=1/3$ we get the following.

\begin{theorem}
  For non-preemptive machine minimization with agreeable deadlines, there is a online $16$-competitive algorithm.
\end{theorem}

\section{Equal processing times}
\label{sec:special-cases2}

Consider the special case of (semi-)online machine minimization of jobs that have equal processing times, that is, $p_j=p$ for all~$j\in J$.

\subsection{Preemptive Scheduling}

We first study preemptive semi-online machine minimization on instances with equal processing times. We firstly show that \EDF is $3$-competitive for semi-online machine minimization. Then we complement this result by a lower bound of $8/7$ on the competitive ratio of any deterministic online algorithm for the special case of equal processing times. Finally, we give a $9.38$-competitive online algorithm.



\begin{theorem}\label{thm:pre-samesize-optknown-EDF}
  \EDF is a $3$-competitive algorithm for semi-online machine minimization when all processing times are equal.
\end{theorem} 
\begin{proof}
Suppose the theorem is not true. Among those instances \EDF fails at, we pick one with the minimum number of jobs. It is easy to see that in this counterexample there is only one job missing its deadline. Let $j$ be this job. Then $d_j$ is the maximum deadline among all the jobs and we let $d_j=d$. Furthermore, during $[d-p,d]$ there must be some time when job $j$ is preempted. Let $t_0\in [d-p,d]$ be the time when job $j$ is preempted, and $S_{t_0}$ be the set of jobs that are processed at time $t_0$.

We claim that the total workload processed by \EDF during $[t_0-p,d]$ is at least $3mp$. To see why, consider the \EDF schedule. For simplicity we assume that $S_{t_0}$ contains job $1$ to job $3m$ and job $i$ is always scheduled on machine $i$ during $[t_0-p,d]$ for $1\le i\le 3m$. We show that the workload of machine $i$ in $[t_0-p,d]$ is at least $p$ and the claim follows. Consider job $i$. 
Then due to the fact that job $i$ does not miss its deadline, we know that the workload on machine $i$ during $[t_0,d]$ is at least $p_i(t_0)$. Consider any time in $[t_0-(p-p_i(t_0)),t_0]$. either job $i$ is scheduled or it is not in favor of other jobs. In both cases, the workload of machine $i$ is at least $p-p_{i}(t_0)$. Thus, the workload of machine $i$ during $[t_0-p,d]$ is at least $p$.

Since \EDF fails, the previous considerations yield 
\begin{equation}\label{eq:unitp1}
  W_{\EDF}(t_0-p)>3mp.
\end{equation}
By the feasibility of the instance however, we know that 
\begin{equation}\label{eq:unitp2}
W_{\OPT}(t_0-p)\le m\cdot[d-(t_0-p)]\le 2mp
\end{equation} (recall that $t_0\ge d-p$). Therefore before time $t_0-p$, there must exist some time $t$ when $w_{\EDF}(t)\le w_{\OPT}(t)\le m$ and, if $t'$ is the latest such time, that $w_{\EDF}(t)>w_{\OPT}(t)$ for any $t'<t<t_0$. As there are at most $m$ active jobs at time $t'$ in the \EDF schedule, we now get that 
\begin{equation}\label{eq:unitp3}
  W_{\EDF}(t')\le W_{\OPT}(t')+mp.
\end{equation}
On the other hand, we have the inequalities $W_{\EDF}(t')\ge W_{\EDF}(t_0-p)+m\cdot(t_0-p-t')$ and $W_{\OPT}(t')\le W_{\OPT}(t_0-p)+m\cdot(t_0-p-t')$. Plugging in Inequalities~\ref{eq:unitp1} and~\ref{eq:unitp2}, respectively, into both of the latter ones yields a contradiction to Inequality~\ref{eq:unitp3}.
\end{proof}

Moreover, we provide a non-trivial lower bound. The underlying structure resembles the one known from the lower bound proof for the general semi-online machine minimization problem from~\cite{phillipsSTW02}.

\begin{theorem}\label{thm:LB}
Let $c<8/7$. On instances fulfilling $p_j\equiv p\in\N$, there does not exist a $c$-competitive algorithm for the semi-online machine minimization problem.
\end{theorem}

W.l.o.g.~we restrict to even $m$. We will heavily make use of the following technical lemma.

\begin{lemma}
Let $c<8/7$, $m$ be even and $\mathsf{A}$ be a $c$-competitive algorithm for semi-online machine minimization.

Assume that $J$ with $p_j=2$, for all $j\in J$, is an instance with the following properties:
\begin{enumerate}[(i)]
  \item $W_{\OPT(J)}(t)=0$,
  \item $W_{\A(J)}(t)\geq w$ and 
  \item $d_j=t+3$ for all jobs $j$ active at $t$ in $\A(J)$.
\end{enumerate}
Then there exists an instance $J^\prime\supset J$ such that $p_j=2$, for all $j\in J'$, with the following properties:
\begin{enumerate}[(i)]
  \item $W_{\OPT(J')}(t+3)=0$,
  \item $W_{\A(J')}(t+3)\geq w^\prime>w$ and 
  \item $d_j=t+6$ for all jobs $j$ active at $t+3$ in $\A(J')$.
\end{enumerate}
\end{lemma}
\begin{proof}
We augment $J$ the following way. At $t$, we release $m$ \textit{more tight} jobs and $m/2$ \textit{more loose} jobs. For each job $j$ of the more tight jobs, we set $r_j=t$ and $d_j=t+3$; for each of the more loose jobs, we set $r_j=t$ and $d_j=t+6$.

We first observe that any $c$-competitive algorithm $\mathsf{A}$ has to do at least $w+2m-2m\cdot(c-1)=w+2m\cdot(2-c)$ work on the remaining active jobs from $J$ as well as the more tight jobs (i.e., all active jobs except for the more loose ones) in the interval $[t,t+2]$. That is because $m$ tight jobs (i.e., $d_j=t+4$ for each job $j$ of them) could be released at $t+2$, forcing any $c$-competitive algorithm at $t+2$ to have left not more work than $2m\cdot(c-1)$ of the jobs due at $t+3$. Also note that in an optimal schedule, all the jobs could be feasibly scheduled.

By the above observation, we get that \A is not $c$-competitive if $w+2m\cdot(2-c)>2cm$. Otherwise, we can upper bound the work done on the more loose jobs in the interval $[t,t+2]$ by $2cm-w+2m\cdot(c-2)=4m\cdot(c-1)-w$, amounting to $4m\cdot(c-1)-w+cm/2=4m\cdot(9c/8-1)-w$ in the interval $[t,t+3]$. What remains (of the more loose jobs) at $t+3$ is thus at least a total workload of $cm-4m\cdot(9c/8-1)+w=4m\cdot(1-7c/8)+w$, which is larger than $w$ if and only if $c<8/7$.
\end{proof}

This allows us to prove the theorem.

\begin{proof}[Proof of Theorem~\ref{thm:LB}.]
We assume there exists a $c$-competitive algorithm \A for semi-online machine minimization. Obviously, the lemma can be applied for $J=\emptyset$ and $t=0$ (note that $w=0$). By re-applying the lemma to the resulting instance for $\lceil 3cm\rceil$ more times, we can be sure that we finally have $W_{\A(J')}(t')>3cm$ for some $J'$ and $t'$ (by the strict increase of this value with each iteration), which is a contradiction.
\end{proof}

Now consider the online instead of semi-online case. Using our doubling technique as a black box (Theorem~\ref{thm: black box}), we are able to derive a  $12$-competitive algorithm. In this subsection we will give an improved $9.38$-competitive algorithm by utilizing the $e$-competitive algorithm of Devanur et
al.~\cite{DevanurMPY14}, which is actually used for the special case when $p=1$.

\begin{theorem}\label{thm:preempt-samesize-optunknown}
There exists a $9.38$-competitive algorithm for the online machine minimization problem when $p_j=p$.
\end{theorem}

As $m$ is unknown, we use the density $\rho(t)$ as an estimation for the number of machines needed for the jobs released until time $t$, where
$$\rho(t)=p\cdot \max_{[a,b]\subseteq[0,d_{max}]}\frac{|\{j\in J(t)\mid [r_j,d_j]\subseteq [a,b]\}|}{b-a}.$$
It can be easily seen that for any interval $[a,b]$, $|\{j\in J(t) \mid r_j\le t, [r_j,d_j]\subseteq [a,b]\}$ represents the set of all the jobs that have been released until time $t$, and have to be finished within $[a,b]$. Thus $\rho(t)$ serves as a lower bound on $m(J(t))$. Using $\rho(t)$, we give the following algorithm (parameterized by $\alpha$ and $c$):

\medskip
\noindent\textbf{Algorithm:} For $\alpha$-tight jobs, run \Earlyfit, i.e., (non-preemptively) schedule a job immediately when it is released. For $\alpha$-loose jobs, run \EDF with $c\rho(t)$ machines.
\medskip

We show in the following that with some parameter $c$ the algorithm gives a feasible schedule.

\begin{lemma}
\Earlyfit uses at most $\lfloor (1/\alpha+1)\rho(t)\rfloor$ machines.
\end{lemma}
\begin{proof}
Suppose the lemma is not true. Then at some time $t$ the Early-Fit algorithm requires more than $\lfloor (\alpha+1)\rho(t)\rfloor$ machines, which implies that by scheduling the $\alpha$-tight jobs immediately, there are more than $(1/\alpha+1)\rho(t)$ tight jobs overlapping at time $t$. Hence each of the $\alpha$-tight job is released after $t-p$ and has a deadline at most $t+p/\alpha$. Therefore, it holds that 
$$\rho(t)> \frac{(1/\alpha+1)\cdot\rho(t)\cdot p}{p/\alpha+p}=\rho(t),$$
which is a contradiction.
\end{proof}

From now on we only consider loose jobs and let $\rho'(t)$ be the density of these jobs. Let $J$ be the set of all loose jobs and $J'$ be the modified instance by replacing each loose job $j\in J$ with $p$ unit size jobs with feasible time window $[r_j,d_j]$. We observe that instance $J$ and $J'$ share the same density $\rho'(t)$. Furthermore, there is an $e$-competitive algorithm for $J'$ which uses $\lceil e\rho(t)\rceil$ machines. Let $W_J(t)$ be the remaining workload of jobs by applying our algorithm to $J$ with $c\ge e$ and $W_{J'}(t)$ be the remaining workload of jobs by applying the $e$-competitive algorithm to $J'$. We have the following load inequality:
\begin{lemma}\label{loadineq}
 For every time $t$ it holds that
 \begin{align*}
   W_J(t) \le W_{J'}(t) +  e \cdot \rho'(t) \cdot p \, .
 \end{align*}
\end{lemma}
\begin{proof}
We prove by induction on $t$. Obviously the lemma is true for $t=0$ as $W_J(0)=W_{J'}(0)$. Suppose the lemma holds for $t\le k$. We prove that $ W_J(k+1) \le W_{J'}(k+1) + e \cdot \rho'(k+1) \cdot p$. 

Consider $w_J(k)$, i.e., the workload processed during $[k,k+1]$ by our algorithm. If $w_J(k)< \lceil e \cdot \rho'(t)\rceil$, then there exists a free machines, implying that there are at most $e \cdot \rho'(t)$ jobs unfinished. Thus, we get $ W_J(k+1) \le W_{J'}(k+1) + e \cdot \rho'(k+1) \cdot p$. Otherwise, we get $w_J(k)\ge \lceil e \cdot \rho'(t)\rceil$. Recall that the lemma is true for $t=k$. That means that we have $W_J(k) \le W_{J'}(k) +  e \cdot \rho'(k) \cdot p $. Since it holds that $W_J(k+1)= W_J(k)-w_J(k)\le W_J(k)-\lceil e \cdot \rho'(k)\rceil$, $W_{J'}(k+1)=W_{J'}(k)-w_{J'}(k)\ge W_{J'}(k)-\lceil e \cdot \rho'(k)\rceil$ as well as $\rho'(k+1)\ge \rho'(k)$, simple calculations show that the lemma is true.
\end{proof}

Using the above load inequality, we prove the following:
\begin{lemma}
For $c\ge e\cdot(1+\alpha)/(1-\alpha)$, \EDF is a $c\rho'(t)$-competitive algorithm for online machine minimization if $p_j\equiv p$ and every job is $\alpha$-loose.
\end{lemma}
\begin{proof}
Among those instances \EDF fails at, we pick one with the minimum number of jobs. It can be easily seen that in such an instance there is only one job missed, and this job has the largest deadline. Let $j$ be the missed job. Then we know that until time $d_j-p/\alpha$, every job is released, implying that $\rho'(t)$ does not change after time $t=d_j-\lfloor p/\alpha\rfloor$. Let $\rho'(t)=\rho'_{max}$ for $t\ge d_j-\lfloor p/\alpha\rfloor$.
Furthermore the fact that job $j$ is missed implies that during $[d_j-\lfloor p/\alpha\rfloor,d_j]$, there must exist at least $\lfloor p/\alpha\rfloor-p+1$ time units when every machine is busy. This means that $$W_J(d_j-p/\alpha)>c\rho'_{max}\cdot(\lfloor p/\alpha\rfloor-p+1)+p-1.$$ 

On the other hand, the load inequality implies that 
$$W_J(d_j-\lfloor p/\alpha\rfloor)\le W_{J'}(d_j-\lfloor p/\alpha\rfloor)+ e \cdot \rho'_{max} \cdot p\le \lfloor p/\alpha\rfloor\cdot\lceil e\rho'_{max}\rceil+e \cdot \rho'_{max}.$$
Simple calculations show that choosing $c\ge e(1+\alpha)/(1-\alpha)$ yields a contradiction.
\end{proof}

We can now prove the theorem.

\begin{proof}[Proof of Theorem~\ref{thm:preempt-samesize-optunknown}.]
Combining the two algorithms for tight and loose jobs, we derive an $(e\cdot(1+\alpha)/(1-\alpha)+ 1/\alpha+1)$-competitive algorithm. By taking $\alpha\approx 0.33$, we get a $9.38$-competitive algorithm. 
\end{proof}

\subsection{Non-preemptive scheduling}

We consider non-preemptive semi-online machine minimization with equal processing times and also show a~$4$-competitive algorithm. 

Let $P=\{0,p,2p,\dots\}$ be the set of all the multiples of the integer $p$. For any job $j$, we let $Q_j=[r_j,d_j]$ be its feasible time window. We know that $Q_j\cap P\neq \emptyset$ since $d_j-r_j\ge p$. We now define job $j$ to be a {\em critical} job if $|Q_j\cap P|=1$, and a {\em non-critical} job if $|Q_j\cap P|\ge2$. 

When we say we round a non-critical job $j$, we round up its release date $r_j$ to be its nearest value in $P$, and round down $d_j$ to be its nearest value in $P$. 
\medskip

\noindent\textbf{Algorithm:} Open $2m$ machines to schedule the non-critical jobs on them. Upon the release of a new job, check if it is critical. If it is critical, schedule it using \Earlyfit, i.e., schedule it immediately upon its release on a separate machine. Otherwise round it and add it to the set of all the non-critical jobs that have not been scheduled. At time $\tau p$ and for any integer $\tau$, pick $2m$ jobs from the set via \EDF and schedule them on the dedicated machines.
\medskip

We first prove the following lemma.

\begin{lemma}\label{lemma:nonpre-samesize-semi}
For every instance $J$, there exists a schedule that uses at most $2m$ machines such that every critical job is scheduled via \Earlyfit, i.e., immediately, and every non-critical job is scheduled exactly in $[\tau p,(\tau+1)\cdot p]$ for some integer $\tau$.
\end{lemma}
\begin{proof}
Consider the (offline) optimum schedule $S$ for the instance $J$. We alter the schedule in the following way. Firstly, we schedule critical jobs immediately when they are released. Secondly, we move any non-critical job which is processed all the time during some $[\tau p-1,\tau p+1]$ (for $\tau$ integer) either completely into $[(\tau-1)\cdot p,\tau p]$ or into $[\tau p,(\tau+1)\cdot p]$. We claim that by doing so we increase the number of machines by at most $m$, which we show for each $\tau$ in the following. 

Consider any critical job $j'$ whose feasible time window intersects with $P$ at point $\tau p$. We observe that in any feasible solution, and hence in the solution $S$, $j'$ is being processed all the time during $[\tau p-1, \tau p+1]$. We let $y_c$ be the number of these jobs. Similarly, we let $y_c'$ be the number of those critical jobs that intersect with $P$ at $(\tau+1)\cdot p$.

Consider non-critical jobs. Let $y_s$ be the number of non-critical jobs that are completely scheduled in $[\tau p,(\tau+1)\cdot p]$ in $S$. Let $y_l$ be the number of non-critical jobs whose feasible time windows contain $[\tau p,(\tau+1)\cdot p]$ and which are processed within $[\tau p,\tau p+1]$ in $S$ but are not completely scheduled in $[\tau p,(\tau+1)\cdot p]$. We know that $$y_s+y_c+y_l\le m(J).$$ Similarly, let $y_r$ be the number of non-critical jobs whose feasible time windows contain $[\tau p,(\tau+1)\cdot p]$ and which are processed within $[(\tau+1)\cdot p,(\tau+1)\cdot p-1]$ in $S$ but are not completely scheduled in $[\tau p,(\tau+1)\cdot p]$. We have $$y_s+y_c'+y_r\le m(J).$$

Now after we alter the solution $S$, the maximum load during $[\tau p,(\tau+1)\cdot p]$ is at most $y_c+y_c'+y_s+y_l+y_r\le 2m(J)$.
\end{proof}

We observe that, in a schedule that satisfies the above lemma, all non-critical jobs could actually be viewed as unit-size jobs. The following lemma can be proved via a simple exchange argument as in~\cite{KaoCRW12}.

\begin{lemma}\label{lem:unit-non-critical}
If there exists a feasible schedule for unit-size jobs that uses $\mu$ machines at time $[t,t+1]$, then applying \EDF to schedule these jobs with $\mu$ machines also returns a feasible solution.
\end{lemma} 

We are now ready to prove the theorem.

\begin{theorem}
For equal processing times, there exists a $4$-competitive algorithm for non-preemptive semi-online machine minimization.
\end{theorem}
\begin{proof}
	Consider the algorithm defined above. Using Lemma~\ref{lemma:nonpre-samesize-semi} and ~\ref{lem:unit-non-critical}, we get that \EDF for non-critical jobs requires at most $2m$ machines. Furthermore, Lemma~\ref{lemma:nonpre-samesize-semi} also implies that $2m$ machines suffice to schedule the critical jobs via \Earlyfit. Thus, we need at most $4m$ machines to schedule the whole instance.
\end{proof}


Using the \double as a black box (Theorem~\ref{thm: black box}) for non-critical jobs, we directly obtain the following result.

\begin{theorem}
  For equal processing times, there exists a $10$-competitive algorithm for non-preemptive online machine minimization.
\end{theorem} 

\noindent\textbf{Remark.} Lemma~\ref{lemma:nonpre-samesize-semi} actually implies a $4$-approximation algorithm for the offline scheduling problem, in which the instance is known in advance and the goal is to minimize the number of machines so as to schedule each job in a feasible way. 
\begin{theorem}\label{coro:offline}
There exists a $4$-approximation algorithm for non-preemptive offline machine minimization when the processing times of all jobs are equal.
\end{theorem}
\begin{proof}
Let $m$ be the (unknown) optimum value. We schedule critical jobs immediately when they are released. This requires at most $2m$ machines by Lemma~\ref{lemma:nonpre-samesize-semi}. On separate machines we schedule non-critical jobs, and after rounding they are viewed as unit-size jobs which could be scheduled feasibly on additional $2m$ machines according to Lemma~\ref{lemma:nonpre-samesize-semi}. In~\cite{KaoCRW12} an optimum algorithm is presented for the offline problem of scheduling unit size jobs, and by applying this algorithm we directly get a feasible schedule for non-critical jobs on at most $2m$ machines. Thus, the overall approximation ratio is $3$. 
\end{proof}
  
\section{Uniform deadlines}
\label{sec:special-cases3}

Consider the special case of (semi-)online machine minimization with a uniform deadline, i.e., every job $j\in J$ has the same deadline $d_j=d$.

\subsection{Preemptive Scheduling}

First look at the special case of preemptive semi-online deadline scheduling in which every job $j\in J$ has the same deadline $d_j=d$. We prove that running \LLF with $m$ machines yields a feasible solution and thus \LLF is 1-competitive. 

\begin{theorem}\label{thm:pre-samedeadline-optknown}
  On instances with a uniform deadline, \LLF is $1$-competitive for preemptive semi-online machine minimization.
\end{theorem}
\begin{proof}
We prove via an exchange argument. Let $S_{\OPT}$ be the feasible offline schedule that uses $m$ machines, and $S_{\LLF}$ be the schedule produced by \LLF (and in $S_{\LLF}$ some jobs may miss their deadlines). 

We show that we can inductively alter $S_{\OPT}$ to $S_{\LLF}$ and thereby maintain feasibility. To this end, suppose $S_{\OPT}$ and $S_{\LLF}$ coincide during $[0,t]$ for $t\ge 0$. Let $j$ be any job which is chosen to be scheduled in $S_{\LLF}$ during $[t,t+1]$ but preempted in $S_{\OPT}$. If there is still a free machine in $S_{\OPT}$, we can simply use one of them to schedule $j$. Otherwise every machine is occupied and there exists some $j'$ so that during $[t,t+1]$ job $j'$ is processed in $S_{\OPT}$ but not in $S_{\LLF}$. Recall that in $S_{\LLF}$ job $j'$ is not processed in favor of job $j$. Thus their laxities satisfy $\ell_{j'}(t)\le \ell_j(t)$, or equivalently, $p_{j'}(t)\ge p_{j}(t)$. Since in $S_{\OPT}$ job $j'$ is chosen to be scheduled in $[t,t+1]$ while job $j$ is preempted, it follows that during $[t+1,d]$ in $S_{\OPT}$, there must exist some time $t'$ such that job $j$ is scheduled, while job $j'$ is not scheduled (either preempted or finished). We now alter $S_{\OPT}$ by swapping the unit of $j'$ processed during $[t,t+1]$ with the unit of $j$ processed during $[t',t'+1]$. By doing so, the resulting schedule is still feasible.
\end{proof}


We conclude this subsection with the following result obtained using Theorem~\ref{thm: black box}.
\begin{theorem}
  On instances with a uniform deadline, there is a $4$-competitive algorithm for preemptive online machine minimization.
\end{theorem} 

\subsection{Non-preemptive Scheduling}

Consider non-preemptive semi-online machine minimization with a uniform deadline.

Recall the definition of $\alpha$-tight and $\alpha$-loose jobs. Since $d_j=d$ in this special case, a job is $\alpha$-tight if $p_j> \alpha(d-r_j)$, and $\alpha$-loose otherwise. We give a $5.25$-competitive algorithm.

\medskip
\noindent\textbf{Algorithm:} We schedule ($\alpha$-)tight and loose jobs separately. 
For tight jobs, we use \Earlyfit, that is, schedule each job immediately when they are released. For loose jobs, we apply \EDF with $m/(1-\alpha)^2$ machines.
\medskip

\noindent\textbf{Remark.} Since jobs have the same deadline, \EDF actually does not distinguish among jobs and thus represents any busy algorithm.

\begin{lemma}\label{lemma:nonpre-deadline-optknown}
\Earlyfit is a $\lceil 1/\alpha\rceil$-competitive algorithm for (semi-)online machine minimization.
\end{lemma}
\begin{proof}
Let $k=\lceil 1/\alpha\rceil$ for simplicity. Then for any tight job $j$ we have $p_j>1/k\cdot(d-r_j)$. Consider the solution produced by \Earlyfit on tight jobs and assume it uses $m'$ machines. There exists some time $t$ such that $w_{\Earlyfit}(t)=m'$. Let $S$ be this set of these $m'$ jobs processed at $t$. We claim that any $k+1$ jobs from $S$ could not be put on the same machine in any feasible solution, yielding that $m\leq m'/k$.

To see why the claim is true, suppose that there exist $k+1$ jobs from $S$ that can be scheduled on one machine, and let them be job $1$ to job $k+1$ for simplicity. We know that $r_j\le t$ for $1\le j\le k+1$, i.e., we get $p_j>1/k\cdot (d-r_j)\ge 1/k\cdot (d-t)$. Now consider the schedule where the $k+1$ jobs are scheduled on one machine, and assume w.l.o.g.~that job $1$ is scheduled at first. Recall that even if job $1$ is scheduled immediately after its release, it will be processed during $[t,t+1]$. Thus job $2$ to job $k+1$, which are scheduled after job $1$, could not start before time $t+1$. Hence, they will be scheduled in $[t+1,d]$. However, $p_j>1/k\cdot(d-t)$, which is a contradiction. 
\end{proof}

We can now prove the following theorem:

\begin{theorem}
  There is a $5.25$-competitive algorithm for non-preemptive semi-online machine minimization with a uniform deadline.
\end{theorem}
\begin{proof}
For loose jobs, we know that the scheduling problem with same deadline jobs is a special case of scheduling with agreeable deadlines jobs. Hence according to Lemma~\ref{lemma:agreeable-loose-edf}, \EDF is a $1/(1-\alpha)^2$-competitive semi-online algorithm. 
Combining the above lemma and setting $\alpha=1/3$, we derive the $5.25$-competitive algorithm.
\end{proof}


Using \double as a black box (Theorem~\ref{thm: black box}), we directly derive a $21$-competitive algorithm for non-preemptive online machine minimization with a uniform deadline. However, since \Earlyfit is already an online algorithm, we only need to apply the doubling technique to transform \EDF into an online algorithm. This yields a $(4/(1-\alpha)^2+\lceil1/\alpha\rceil)$-competitive algorithm. This performance guarantee is optimized for $\alpha=1/4$, yielding the following result.

\begin{theorem}
  There is a $11\frac{1}{9}$-competitive algorithm for non-preemptive online machine minimization with a uniform deadline.
\end{theorem}

\section{A Preemptive  $\mathcal{O}(\log n)$-competitive Algorithm}
\label{sec:general logn}

We give a preemptive $\mathcal{O}(\log n)$-competitive semi-online
algorithm and apply \double (Theorem~\ref{thm: black box}). Assume for
the remainder of this section that the optimal number of machines is known.

\subsection{Description}

We fix some parameter $\alpha\in(0,1)$. At any time $t$, our algorithm maintains a partition of $J(t)$ into  {\em safe} and {\em critical} jobs (at $t$), which we define as
\[L_t=L_{t-1}\cup\left\{j\in J(t)\mid p_j(t)\leq\alpha(d_j-t)\right\}\;(L_{-1}=\emptyset)\textup{ and }T_t=J(t)\setminus L_t, \textup{ respectively.}\]
In other words, once a job gets $\alpha$-loose at some time, it gets classified as safe job from then on and is only considered as critical if it has {\em always} been $\alpha$-tight.

Our algorithm does the following. Once a job is classified as safe it is taken care of by $\EDF$ on separate machines. More formally, consider $L=\bigcup_t\left\{j_t\mid j\in L_t\setminus L_{t-1}\right\},$ where $j_t$ is the {\em residue} of job $j$ at time $t$, i.e., it is a job with processing time $p_j(t)$ and  time window $[t,d_j]$. We run \EDF on $L$ using $m(L)/(1-\alpha)^2$ machines. The critical jobs (at time $t$) require a more sophisticated treatment on $\Theta(m\log|J(t)|)$ machines, which we describe below. 

Consider the set of critical jobs $T_t$ at time $t$. We partition $T_t$ into $h_t\cdot\mu_t$ different subsets, each of which is processed by \EDF on a separate machine. If $t$ is not a release date, we simply take over the partition from $t-1$ and remove all the jobs that have become safe since then. Otherwise, assume w.l.o.g.~that $T_t=\{1,\dots,|T_t|\}$ with $d_j\geq d_{j+1}$, for all $j$, and first partition $T_t$ into $h_t$ many sets $S_1,\dots,S_{h_t}$. We consider the jobs in increasing order of indices and assign them to existing sets or create new ones, as follows. 
Let $a(i)$ be the job with the earliest deadline in an existing set $S_i$. Then a job $j$ is added to an arbitrary set $S_i$ such that the (remaining) length of the feasible time window of $j$ is smaller than the laxity of $j_i$, i.e., $d_j-t\leq \ell_{a(i)}(t)$. If no such set exists, we create a new set and add $j$.

\begin{footnotesize}
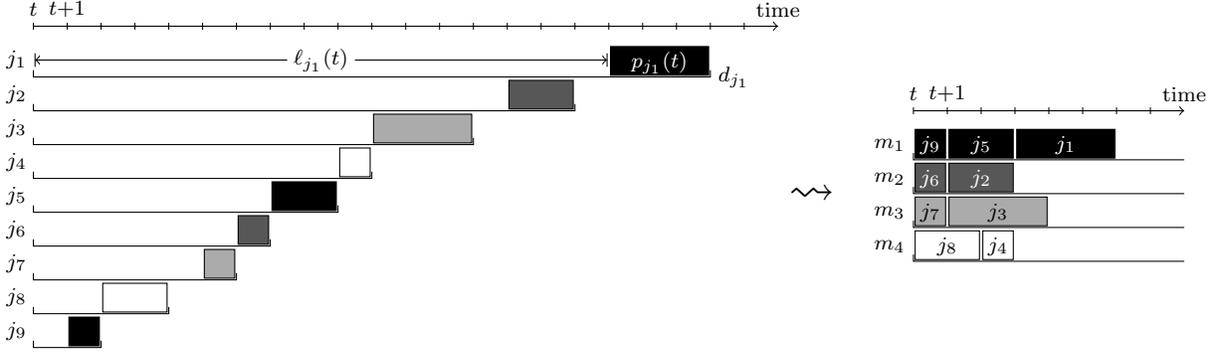
\begin{figure}
\centering
\begin{tikzpicture}[scale=0.45]
  \newcounter{x}
  \newcommand\job[3]{
    \addtocounter{x}{1}
    \begin{scope}[shift={(0,-\value{x})}]
      \filldraw[fill=black!#3, shift={(-0.05,0)}] (#1.1,0.05) rectangle (#2,0.9);
      \draw (0,0.2) -- (0,0) -- (#2,0) -- +(0,0) -- +(0,0.2);
      \node at (-0.5,0.5) {\scriptsize{$j_\text{\arabic{x}}$}};
    \end{scope}
  }
  \job{17}{20}{100}
  \job{14}{16}{66}
  \job{10}{13}{33}  
  \job{9}{10}{0}  
  \job{7}{9}{100}  
  \job{6}{7}{66} 
  \job{5}{6}{33} 
  \job{2}{4}{0}
  \job{1}{2}{100}
  \node at (23,-4.5) {\LARGE{$\leadsto$}};  
  \draw[->] (0,0.5) -- (22,0.5);
  \foreach \tick in {0,...,21}
    \draw (\tick,0.4) -- (\tick,0.6);
  \node at (0,1) {\scriptsize{$t$}};
  \node at (1,1) {\scriptsize{$t$+$1$}};     
  \node at (22,1) {\scriptsize{time}};  
  \node at (20.7,-1) {\scriptsize{$d_{j_1}$}};
  \node at (18.5,-0.57) {\textcolor{white}{\scriptsize{$p_{j_1}(t)$}}};
  \draw[<-] (0.05,-0.5) -- (7.5,-0.5);
  \draw[->] (9.5,-0.5) -- (16.95,-0.5);
  \draw (0.05,-0.3) -- (0.05,-0.7);
  \draw (16.95,-0.3) -- (16.95,-0.7);
  \node at (8.5,-0.5) {\scriptsize{$\ell_{j_1}(t)$}}; 
  
  \newcommand\sjob[7]{
      \filldraw[fill=black!#3, shift={(-0.05,0)}] (#1.1,-#4.05) rectangle (#2,-#4.9);
      \node at (#5,-#4.5) {\textcolor{black!#7}{\scriptsize{$j_#6$}}};
  }
  \begin{scope}[shift={(26,-2.5)}]
    \draw[->] (0,0.5) -- (8,0.5);
    \foreach \tick in {0,...,7}
      \draw (\tick,0.4) -- (\tick,0.6);
    \foreach \y/\machine in {0/1,1/2,2/3,3/4}{   
      \node at (-0.7,-\y.5) {\scriptsize{$m_\machine$}}; 
      \draw[shift={(0,0.05)}] (0,-\y.8) -- (0,-\machine) -- (8,-\machine);
    }
    \node at (0,1) {\scriptsize{$t$}};
    \node at (1,1) {\scriptsize{$t$+$1$}};     
    \node at (8,1) {\scriptsize{time}};        
    \sjob{0}{1}{100}{0}{0.5}{9}{0}
    \sjob{1}{3}{100}{0}{2}{5}{0}
    \sjob{3}{6}{100}{0}{4.5}{1}{0}
    \sjob{0}{1}{66}{1}{0.5}{6}{0}
    \sjob{1}{3}{66}{1}{2}{2}{0}
    \sjob{0}{1}{33}{2}{0.5}{7}{100}
    \sjob{1}{4}{33}{2}{2.5}{3}{100}
    \sjob{0}{2}{0}{3}{1}{8}{100}
    \sjob{2}{3}{0}{3}{2.5}{4}{100}            
  \end{scope}            
\end{tikzpicture}
\caption{Visualization of jobs in set $S_i\subset T_t$ and their  (temporal) distribution over machines. Left: Time window of each job with its processing volume pushed to the right. Right: resulting job-to-machine~assignment. Assume $\alpha$ to be small enough that none of the jobs gets safe before completion as well as $\mu_t=4$ (the machines are called $m_1,\dots,m_4$).
}
\label{mp-fig: main-alg}
\end{figure}
\end{footnotesize}

Consider an arbitrary $S_i$. By construction, we could feasibly schedule this set of jobs on a single machine using $\EDF$. However, doing so may completely use up the laxity of some job in $S_i$, causing problems when new jobs are released in the future. We want to keep a constant fraction $\beta\in(0,1)$ (independent of $t$) of the initial laxity $\ell_j(r_j)$ for each $j\in S_i$ by distributing them carefully over $\mu_t>1$ machines. Consider again $S_i=\{j_1,\dots,j_{|S_i|}\}$ and assume an increasing order by deadlines. We further partition $S_i$ into $\mu_t$ different subsets $S_i^1,\dots,S_i^{\mu_t}$ by setting $S_i^k=\{j_i\mid i\equiv k-1 \mod \mu_t\}$ and schedule {\em each} of them on a separate machine via $\EDF$. By the right choice of $\mu_t$, this reduces the decrease in laxity sufficiently. The procedure is visualized in Figure~\ref{mp-fig: main-alg}.

We know from the preliminaries that \EDF on $L$ requires $m(L)/(1-\alpha)^2$ machines. Hence the algorithm uses $\OO(m(L)+\mu_t\cdot h_t)$ machines. In Subsection~\ref{subsec: oa-analysis}, we prove that we can choose $\mu_t=\mathcal{O}(\log|J(t)|)$ so that a constant $\beta$ as above exists and that $h_t=\mathcal{O}(m)$, $m(L)=\OO(m)$. The key to prove that lies in Subsection~\ref{subsec: oa-pow-lax} where we characterize the relation between a decrease in the laxity and increase in the number of machines. Hereby, we use the compact representation of the optimum value presented in Subsection~\ref{subsec:optimum}. In summary we need $\mathcal{O}(m)+\mathcal{O}(m)\cdot\mathcal{O}(\log|J(t)|)=\mathcal{O}(\log n)\cdot m$ machines. 

\subsection{Strong Density -- A Compact Representation of the Optimum}
\label{subsec:optimum}

The key to the design and analysis of online algorithms are good lower bounds on the optimum value. The offline problem is known to be solvable optimally in polynomial time, e.g., by solving a linear program or a maximum flow problem~\cite{horn74}. However, these techniques do not provide a quantification of the change in the schedule and the required number of machines when the job set changes. We derive an exact estimate of the optimum number of machines as an expression of workload that must be processed during some (not necessarily consecutive) intervals.



Let~$\II$ denote a set of $k$ pairwise disjoint intervals~$[a_h,b_h], 1\leq h\leq k$. We define the {\em length} of $\II$ to be $|\II|=\sum_{1\leq h\leq k}|b_h-a_h|$ and its {\em union} to be $\cup\II=\bigcup_{1\leq h\leq k}[a_h,b_h]$. 

\begin{definition}[Strong density]
  Let $\II$ be as above and $Q_j=[r_j,d_j]$ be the feasible time window of job $j$. The contribution $\Gamma_j(\II)$ of job~$j$ to $\II$ is the least processing volume of~$j$ that must be scheduled within $\cup\II$ in any feasible schedule, i.e.,
$$\Gamma_j(\II)=\max\{0,|(\cup\II)\cap Q_j|-\ell_j\}=\max\{0,|(\cup\II)\cap Q_j|-d_j+r_j+p_j)\}.$$
The strong density $\rho_s$ is defined as maximum total contribution of all jobs over all interval sets:
$$\rho_s=\max_{\II\subseteq\{[t,t+1]\mid 0\leq t<d_{\max}\}} \frac{\sum_{j=1}^n\Gamma_j(\II)}{|\II|}.$$
\end{definition}



Our main result in this section is that $\lceil \rho_s\rceil$ is the exact value of an optimal solution. We give a combinatorial proof, but we also note that it follows from LP duality.
\begin{theorem}
\label{thm: strong-density}
$m=\lceil\rho_s\rceil.$
\end{theorem}
\begin{proof}
It is easy to see that $\lceil\rho_s\rceil$ is a lower bound on $m$ since the volume of $\sum_{j=1}^n\Gamma_j(\II)$ must be scheduled in $\cup\II$ of length $|\II|$. This yields a lower bound of $\lceil \sum_j \Gamma_j(\II)/|\II| \rceil$ for any $\II$.

It remains to show that $m\le \lceil\rho_s\rceil$. Given a schedule, we denote by $\chi_h$ the number of time slots $[t,t+1]$ during which exactly $h$ machines are occupied. For any feasible schedule that uses $m$ machines, we obtain a vector $\chi=(\chi_m,\cdots, \chi_0)$. We define a lexicographical order $<$ on these vectors: we have $\chi<\chi'$ if and only if there exists some $h$ with $\chi_h<\chi_h'$ and $\chi_i=\chi_i'$, for all $i<h$. We pick the schedule whose corresponding vector is the smallest with respect to $<$.

We now construct a directed graph $G=(V,A)$ based on the schedule we pick. Let $V=\{v_0,\dots,v_{d_{\max}-1}\}$ where $v_t$ represents the slot $[t,t+1]$.
Let $\phi_t$ be the number of machines occupied during $[t,t+1]$. An arc from $v_i$ to $v_k$ exists iff $\phi_i\ge \phi_k$ and there exists some job $j$ with $[i,i+1]\cup [k,k+1]\subseteq Q_j$ whereas $j$ is processed in $[i,i+1]$ but not in $[k,k+1]$.
Intuitively, an arc from $v_i$ to $v_k$ implies that one unit of the workload in $[i,i+1]$ could be carried to $[k,k+1]$.

We claim that in $G$ there is no (directed) path which starts from some $v_i$ with $\phi_i=m$, and ends at some $v_\ell$ with $\phi_\ell\le m-2$. Suppose there exists such a path, say $(v_{i_1},\dots, v_{i_\ell})$ with $\phi_{i_1}=m$ and $\phi_{i_\ell}\le m-2$. Then we alter the schedule such that we move one unit of the workload from $[i_s,i_s+1]$ to $[i_{s+1},i_{s+1}+1]$, for all $s<\ell$. By doing so, $\phi_{i_1}$ decreased and $\phi_{i_\ell}$ increased by $1$ each. By $\phi_{i_1}=m$ and $\phi_{i_\ell}\le m-2$, we get that $\chi_m$ decreases by $1$, contradicting the choice of the schedule.

Consider $V_m=\{v_i\mid\phi_i=m\}$ and let $P(V_m)=\{v_{i_1},\dots,v_{i_\ell}\}$ be the set of vertices reachable from $V_m$ via a directed path (trivially, $V_m\subseteq P(V_m)$). The above arguments imply that for any $v_i\in P(V_m)$, it holds that $\phi_i\ge m-1$. We claim that $\sum_j\Gamma_j(\II)=\sum_h \phi_{i_h}$ for $\II=\{[i_h,i_h+1]|1\le h\le \ell\}$. Note there is no arc going out from $P(V_m)$ (otherwise the endpoint of the arc is also reachable and would have been included in $P(V_m)$), i.e., we cannot move out any workload from $\cup\II$, meaning that the contribution of all jobs to $\II$ is included in the current workload in $\cup\II$, i.e., $\sum_h \phi_{i_h}$. Thus, $$\rho_s\ge \frac{\sum_h\phi_{i_h}}{|\II|}=\frac{\sum_h\phi_{i_h}}{|P(V_m)|}.$$
Notice that $\phi_{i_h}$ is either $m$ or $m-1$, and among them there is at least one of value $m$, thus the right-hand side is strictly larger than $m-1$, which implies that $\lceil \rho_s\rceil \ge m$.  
\end{proof} 

\subsection{The Power of Laxity} \label{subsec: oa-pow-lax}

Let $J$ be an arbitrary instance and $\beta\in (0,1)$ be rational. We analyze the influence that  the decrease in the laxity of each job by a factor of $\beta$ has on the optimal number of machines. To this end, let $J_\beta=\{j_\beta\mid j\in J\}$ be the modification of $J$ where we increase the processing time of each job by a $(1-\beta)$-fraction of its initial laxity, i.e., $[r_{j_\beta},d_{j_\beta}]=[r_j,d_j]$ and $p_{j_\beta}=p_j+(1-\beta)(d_j-r_j-p_j)$ (or equivalently $\ell_{j_\beta}=\beta\cdot\ell_j$). Using a standard scaling argument, we may assume w.l.o.g.~that for a fixed $\beta$ all considered parameters are integral.

Obviously, we have $m(J_\beta)\geq m(J)$. In the following, we show $m(J_\beta)\le\mathcal{O}(1/\beta)\cdot m(J)$ for instances $J$ only consisting of sufficiently tight jobs.

\begin{theorem}\label{thm:laxity-drop} For every $J$ consisting only of $(1/2)$-tight jobs and $\beta\in (0,1)$, it holds that
$$m(J_{\beta})\leq\frac{4}{\beta}\cdot m(J).$$
\end{theorem}

In fact, we even show a slightly stronger but less natural result. More specifically, we split each job $j_\beta\in J_\beta$ into two subjobs: the \textit{left part} of $j_\beta$, denoted $j_\beta^\triangleleft$, and the \textit{right part} of $j$, denoted $j_\beta^\triangleright$, where 
\begin{align*}
&[r_{j_\beta^\triangleleft},d_{j_\beta^\triangleleft}]=\left[r_j,r_j+\left(1-\frac{\beta}{2}\right)\cdot\ell_j\right],\,p_{j_\beta^\triangleleft}=(1-\beta)\cdot\ell_j\\
\text{ and }&[r_{j_\beta^\triangleright},d_{j_\beta^\triangleright}]=\left[r_j+\left(1-\frac{\beta}{2}\right)\cdot\ell_j,d_j\right],\,p_{j_\beta^\triangleright}=p_j,
\end{align*}
i.e., we split the former feasible time window at $r_j+(1-\beta/2)\cdot\ell_j$ and distribute $p_{j_\beta}$ over both emerging jobs in such a way that the right part has the processing time of the original job $j$. The sets containing all of these subjobs are $J^\triangleleft_\beta=\{j^\triangleleft_\beta\mid j\in J_\beta\}$ and $J^\triangleright_\beta=\{j^\triangleright_\beta\mid j\in J_\beta\}$, respectively. In the following, we show Theorem~\ref{thm:laxity-drop} even for $J^\triangleleft_\beta\cup J^\triangleright_\beta$ instead of $J_\beta$. This implies, as can be easily seen, the statement for $J_\beta$.

In the proof, we define two more jobs based on $j$ and $\gamma\in(0,1)$, namely the {\em $\gamma$-right-shortened} and the {\em $\gamma$-left-shortened} variant of $j$, denoted $j^{\leftarrow\gamma}$ and $j^{\gamma\rightarrow}$ and contained in the sets $J^{\leftarrow\gamma}$ and $J^{\gamma\rightarrow}$, respectively. We define
\begin{align*}
&[r_{j^{\leftarrow\gamma}},d_{j^{\leftarrow\gamma}}]=[r_j,d_j-(1-\gamma)\cdot\ell_j],\,p_{j^{\leftarrow\gamma}}=p_j\\
\text{ and }&[r_{j^{\gamma\rightarrow}},d_{j^{\gamma\rightarrow}}]=[r_j+(1-\gamma)\cdot\ell_j,d_j],\,p_{j^{\gamma\rightarrow}}=p_j,
\end{align*}
which means that we drop an amount of $(1-\gamma)\cdot\ell_j$ from either side of the feasible time window while keeping the processing time. 
We then show the following lemma.

\begin{figure}
\centering
\definecolor{light-gray}{gray}{0.7}
\begin{tikzpicture}[scale=0.35]
  \newcommand\poljob[5]{
    \begin{scope}[shift={(#3,#4)}]
      \filldraw[fill=black!70, shift={(-0.05,0)}] (#1,0.05) rectangle (#2,0.9);
      \draw (0,0.2) -- (0,0) -- (#2,0) -- +(0,0) -- +(0,0.2);
      \node at (-0.5,0.5) {\tiny{#5}};
    \end{scope}
  }
  \newcommand\mpoljob[6]{
    \begin{scope}[shift={(#3,#4)}]
      \draw[draw=black!30] (0,0.2) -- (0,0) -- (10,0) -- +(0,0) -- +(0,0.2);    
      \filldraw[fill=black!70, shift={(-0.05,0)}] (#1,0.05) rectangle (#2,0.9);
      \draw (#6,0.2) -- (#6,0) -- (#2,0) -- +(0,0) -- +(0,0.2);
      \node at (-0.85,0.5) {\tiny{#5}};
      \node at (10,-0.5) {\textcolor{light-gray}{\tiny{$d_j$}}};
      \node at (0,-0.5) {\textcolor{light-gray}{\tiny{$r_j$}}};        
    \end{scope}
  }  
  \poljob{6.1}{10}{0}{0}{$j$}
  \poljob{3.1}{10}{0}{-4}{$j_\beta$}
  \mpoljob{1.6}{4.5}{-12}{-8}{$j_\beta^\triangleleft$}{0}
  \mpoljob{6.1}{10}{12}{-8}{$j_\beta^\triangleright$}{4.5}
  \mpoljob{1.6}{5.5}{-12}{4}{$j^{\leftarrow\gamma}$}{0}
  \mpoljob{6.1}{10}{12}{4}{$j^{\gamma\rightarrow}$}{4.5}

  \draw[<-] (0.05,0.5) -- (2,0.5);
  \draw[->] (4,0.5) -- (5.95,0.5);
  \draw (0.05,0.7) -- (0.05,0.3);
  \draw (5.95,0.7) -- (5.95,0.3);
  \node at (3,0.5) {\tiny{$\ell_{j}(t)$}};     
  \node at (8,0.43) {\textcolor{white}{\tiny{$p_{j}(t)$}}};
  \node at (10,-0.5) {\tiny{$d_{j}$}};
  \node at (0,-0.5) {\tiny{$r_{j}$}};
  
  \draw[->, decorate, decoration={snake, amplitude=1.5pt}] (5,-0.2) -- (5,-2.9);
  \node[align=center] at (6.5,-1.5) {\tiny \begin{tabular}{c}drop\\ laxity\end{tabular}};
  
  \draw (5,-4.2) edge[out=270,in=180,->,looseness = 0.8] (10.5,-7.5);
  \draw (5,-4.2) edge[out=270,in=0,->,looseness = 0.8] (-0.5,-7.5);
  \node at (5,-6.5) {\tiny{split}};
   
  \draw[->] (5,1.2) -- (10.5,3.5);
  \node[align=center] at (9.5,2) {\tiny shorten left};  
  
  \draw[->] (5,1.2) -- (-0.5,3.5);
  \node[align=center] at (0.1,2) {\tiny shorten right};    
  
  \draw[->,dashed] (-10,3) -- (-10,-6);
  \node[align=center] at (-12,-2) {\tiny \begin{tabular}{c}strictly\\ harder\\ for $\beta=2\gamma$\end{tabular}};  
  
  \draw[<->,dashed] (19,3) -- (19,-6);
  \node[align=center] at (21,-2) {\tiny \begin{tabular}{c}identical\\ for $\beta=2\gamma$\end{tabular}};             
\end{tikzpicture}
\caption{The different types of jobs defined for the proof of Theorem~\ref{thm:laxity-drop} and their relations. We let $\gamma=\beta/2=1/4$ and $p_j=2\cdot(d_j-r_j)/5$}
\label{fig: power-lax}
\end{figure}
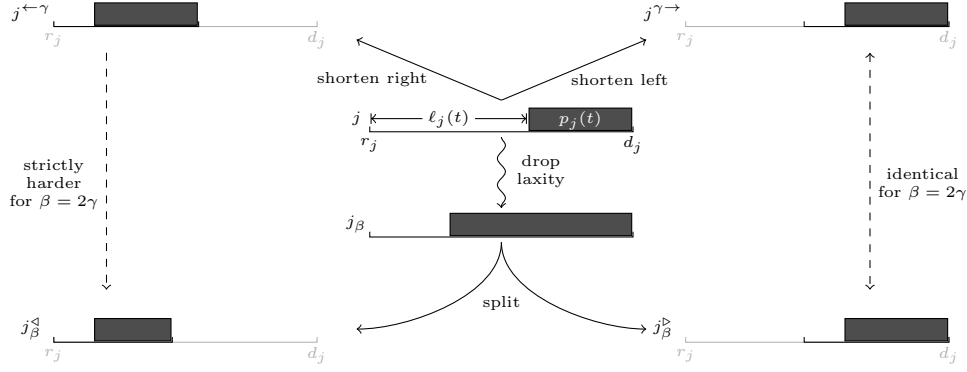

\begin{lemma}\label{lemma:laxity-release-increase}
For every $J$ and $\gamma\in (0,1)$, it holds that
$$m(J^{\leftarrow\gamma})\leq \frac{1}{\gamma}\cdot m(J) \textup{ and } m(J^{\gamma\rightarrow})\leq\frac{1}{\gamma}\cdot m(J).$$
\end{lemma}
\begin{proof}
We representatively show the statement for $J^{\gamma\rightarrow}$; the argumentation for $J^{\leftarrow\gamma}$ is analogous. According to Theorem~\ref{thm: strong-density}, there exists a set of $k$ pairwise disjoint intervals $\II$ such that $$m(J^{\gamma\rightarrow})=\sum_{j^{\gamma\rightarrow}}\frac{\Gamma_{j^{\gamma\rightarrow}}(\II)}{|\II|}.$$
W.l.o.g, we assume that $\II=\{[a_h,b_h]\mid 1\leq h\leq k\}$ with $I_h$ where $a_h<a_{h+1}$ (hence, $b_h<b_{h+1}$), for all $h$. To derive the relationship between $m(J^{\gamma\rightarrow})$ and $m(J)$, we expand $\II$ into a set of element-wise larger intervals $\ex(\II)$ such that $|\II|\geq|\ex(\II)|/\gamma$. We show for each job $j$ that $\Gamma_j(\ex(\II))\geq\Gamma_{j^{\gamma\rightarrow}}(\II)$ and thus the lemma follows.

The expanding works as follows. Given $\II$ as above, we expand each of the intervals $[a_h,b_h]$ to $[a_h',b_h]$ with $a_h'\leq a_h$ the following way. We start at the rightmost interval $I_k$ and try to set $a_k'=b_k-(b_k-a_k)/\gamma$ for $I_k'$. If this would, however, produce an overlap between $I_k'$ and $I_{k-1}$, we set $a_{k}'=b_{k-1}$, $\delta_k=b_{k-1}-(b_k-(b_k-a_k)/\gamma)$ and try to additionally expand $I_{k-1}$ by $\delta_k$ instead. After that, we continue this procedure to the left but never expand an interval into negative time. More formally, we let~$b_0=\delta_{k+1}=0$  and set for all $h$ 
\[a_h'=\max\left\{b_{h-1},b_h-\left(\frac{b_h-a_h}{\gamma}+\delta_{h-1}\right)\right\}\text{ and }\delta_h=\max\left\{0,b_{h-1}-\left(b_h-(\frac{b_h-a_h}{\gamma}+\delta_{h-1})\right)\right\}. \]
The following two facts can be easily observed.

\begin{observation}\label{lemma:expand-length}
Consider $\II$ as above. One of the following statements is true:
\begin{enumerate}[(i)]
  \item We have $\cup\ex(\II)=[0,b_k]$ and $\left|\II\right|\geq|\ex(\II)|/\gamma$.
  \item We have $\cup\ex(\II)\subseteq \left[a_1-\frac{|\II|(1-\gamma)}{\gamma},b_k\right]$ and $|\II|=|\ex(\II)|/\gamma$.
\end{enumerate}
\end{observation}

\begin{observation}\label{lemma:expand-subintervals}
If $\cup\II\subseteq\cup\II'$, then $\cup\ex(\II)\subseteq\cup\ex(\II')$.
\end{observation}

We now consider any job $j^{\gamma\rightarrow}$ with $\Gamma_{j^{\gamma\rightarrow}}(\II)>0$ and claim that $\Gamma_j(\ex(\II))\geq\Gamma_{j^{\gamma\rightarrow}}(\II)$. First plug in the definition of contribution. We get
\begin{align*}
\Gamma_{j}(\ex(\II))&=\max\{0,|(\cup\ex(\II))\cap Q_{j}|-\ell_j\}\\
\text{and}\hspace{0.4cm}\Gamma_{j^{\gamma\rightarrow}}(\II)&=\max\{0,|(\cup\II)\cap Q_{j^{\gamma\rightarrow}}|-\ell_{j'}\}=\max\{0,|(\cup\II)\cap Q_{j^{\gamma\rightarrow}}|-\gamma \ell_{j}\}.
\end{align*}
It thus suffices to prove that $|(\cup\ex(\II))\cap Q_{j}|-(1-\gamma)\cdot\ell_j\ge |(\cup\II)\cap Q_{j^{\gamma\rightarrow}}|$.

Let $\II'=\{I\cap Q_{j^{\gamma\rightarrow}}\mid I\in\II\}$ 
where we let $a'$ be the leftmost point of it and $b'$ be the rightmost one. By Observation~\ref{lemma:expand-subintervals}, it follows that $\cup\ex(\II')\subseteq \cup\ex(\II)$, i.e., we can restrict to proving $$|(\cup\ex(\II'))\cap Q_{j}|-(1-\gamma)\cdot\ell_j\ge |\II'|.$$ According to Observation~\ref{lemma:expand-length}, there are two possibilities. 

\noindent\textbf{Case 1.} We have $\cup\ex(\II')=[0,b']$. Given the fact that $Q_j=Q_{j^{\gamma\rightarrow}}\cup [r_j,r_j+(1-\gamma)\cdot\ell_j]$ and $r_j<b'$, it follows that $|(\cup\ex(\II'))\cap Q_{j}|\ge |\II'|+(1-\gamma)\cdot\ell_j$. 

\noindent\textbf{Case 2.} It holds that $\ex(\II')$ takes up a length of $|\II'|/\gamma$ in the interval $$[a'',b']=\left[a'-\frac{|\II'|(1-\gamma)}{\gamma},b'\right].$$ Hence for any $x\ge0$, it takes up a length of at least $|\II'|/\gamma-x$ in the interval $[a''+x,b']$. Using $\Gamma_{j^{\gamma\rightarrow}}(\II)=\Gamma_{j^{\gamma\rightarrow}}(\II')>0$, we get $|\II'|>\ell_{j^{\gamma\rightarrow}}=\gamma \ell_j$. Thus, if we let $x=(|\II'|-\gamma \ell_j)(1-\gamma)/\gamma$, we get that $\ex(\II')$ takes up a length of at least $|\II'|+(1-\gamma)\cdot\ell_j$ in the interval $[a''+x,b']=[a'+(1-\gamma)\cdot\ell_j,b']$. Since $[a'+(1-\gamma)\cdot\ell_j,b']\subseteq Q_j$, we have $|(\cup\ex(\II'))\cap Q_{j}|\ge |\II'|+(1-\gamma)\cdot\ell_j$, which completes the proof.
\end{proof}


\begin{proof}[Proof of Theorem~\ref{thm:laxity-drop}.]
We show that $m(J^\triangleleft_\beta\cup J^\triangleright_\beta)\leq 4\cdot m(J)/\beta$ or, more specifically, that $m(J^\triangleleft_\beta)=m(J^\triangleright_\beta)=2\cdot m(J)/\beta$. 

For the first part, compare the two instances $J^\triangleleft_\beta$ and $J^{\leftarrow \beta/2}$. We observe that $$[r_{j^\triangleleft_\beta},d_{j^\triangleleft_\beta}]\subseteq[r_{j^{\leftarrow \beta/2}},d_{j^{\leftarrow \beta/2}}]\;\text{ and }\;\ell_{j^\triangleleft_\beta}=\ell_{j^{\leftarrow \beta/2}},$$ for all $j\in J$, where we use for the first part that $j$ is $(1/2)$-tight. This implies that a schedule of $J^{\leftarrow \beta/2}$ on $m'$ machines can be easily transformed into a schedule of $J^\triangleleft_\beta$ on $m'$ machines. Hence, we get that $m(J^\triangleleft_\beta)\leq m(J^{\leftarrow \beta/2})$ and the claim follows by Lemma~\ref{lemma:laxity-release-increase}.

For the second part, we simply observe that $J^\triangleright_\beta=J^{\beta/2\rightarrow}$ and apply Lemma~\ref{lemma:laxity-release-increase} again.
\end{proof}

\subsection{Analysis} \label{subsec: oa-analysis}

For the sake of simplicity, we choose $\alpha\in(0,1/2]$ such that $1/\alpha$ is integral. We now analyze the performance of our algorithm. The goal of this subsection is proving the following theorem.

\begin{theorem}\label{thm:logn}
Our algorithm is an $\OO(\log n)$-competitive algorithm for the machine minimization problem.
\end{theorem}

Recall that we apply \EDF to schedule the safe jobs $L=\bigcup_t\left\{j_t\mid j\in L_t\setminus L_{t-1}\right\}$. Using Theorem~\ref{thm: EDF-small}, $m(L)/(1-\alpha)^2=\OO(m(L))$ machines are suffice to feasibly schedule all the safe jobs. Meanwhile at any time $[t,t+1]$, we partition $T_t$ into $h_t$ subsets $S_1,\dots,S_{h_t}$ and schedule $\mu_t$ jobs of each subset on a separate set of machines.

We first prove that, if we choose a certain $\mu_t=\OO(\log |J(t)|)$, the laxity of each job in $T_t$ never drops below a constant fraction of its original laxity. This implies that our algorithm never misses a critical job. Hence our algorithm does not miss any job.

\begin{lemma}\label{lemma:length-preempted}
We can choose $\mu_t=\OO(\log |J(t)|)$ such that at any time $t$ and for any $j\in T_t$, $\ell_j(t)\ge \beta\ell_j$ where $\beta=2-\pi^2/6\in(0,1)$.
\end{lemma}
\begin{proof}
Let $t'\le t$ be an arbitrary time when the partition of $T_{t'}$ into the sets $S_1,\dots,S_{h_{t'}}$ is recomputed. Let $S_i\ni j$ and w.l.o.g. $S_i=\{1,\dots,|S_i|\}$ where $d_{j'}\geq d_{j'+1}$, for all $j'$. By the way the partition is constructed, we further have $d_{j'+1}-t\leq\ell_{j'}(t')$. Now let $j^\star=j+\mu_{t'}$ and notice that, unless the partition is recomputed, the execution of $j$ is only blocked by other jobs for at most $d_{j^\star}-t'$ (if $j^\star$ does not exist, $j$ is immediately executed). Since every job in $T_{t'}$ is $\alpha$-tight at $t'$, this however means that the execution of $j$ is only delayed by $(1-\alpha)^{\mu_{t'}}\cdot\ell_j(t')\leq (1-\alpha)^{\mu_{t'}}\cdot\ell_j$. Note that we can choose $\mu_{t'}=\OO(\log|J(t')|)$ in such a way that $(1-\alpha)^{\mu_{t'}}\cdot\ell_j\leq \ell_j/|J(t')|^2$. To get the total time $j$ is at most delayed by, we have to sum over all such $t'$. Recall that at least one job is released at any such time $t'$, hence $|J(t')|\ge |J(t'-1)|+1$ and we have: $$\ell_j(t)\geq \ell_j-\sum_{i=2}^{|J(t)|}\frac{\ell_j}{i^2}\geq \ell_j-(\pi^2/6-1)\cdot\ell_j = (2-\pi^2/6)\cdot\ell_j.$$\end{proof}

Before proving the main theorem, we need an additional bound on $h_t$. To this end, let $\hat{T}_t=\{j_t\mid j\in T_t\}$ be the set of residues of $T_t$ at $t$, i.e., the jobs in $T_t$ with remaining processing times and time windows.

\begin{lemma}\label{lemma:partition-at-time-t} At all times $t$, it holds that
$$h_{t}\le 1+\left(2+\frac{2}{\alpha}\right)\cdot m(\hat{T}_t).$$
\end{lemma}

W.l.o.g., we let $\hat{T}_t=\{1,\dots,|\hat{T}_t|\}$ where $d_j\geq d_{j+1}$, for all $j$. We first prove the following auxiliary lemma:

\begin{lemma}\label{lemma:length-shrink}
For any $j$ and $j'=j+2m(\hat{T}_t)/\alpha$, we have $d_{j'}-t\le (d_j-t)/2$.
\end{lemma}
\begin{proof}
Suppose there exists some $j$ such that $d_{j+h}-t>(d_j-t)/2$ holds for any h with $0\le h\le 2m(T_t)/\alpha$. Given that every job in $T_t$ is $\alpha$-tight at $t$, we get $p_{j+h}>\alpha/2\cdot (d_j-t)$. Thus, the total workload that has to be finished in the interval $[t,d_j]$ is larger than $\alpha/2\cdot(d_j-t)\cdot 2m(\hat{T}_t)/\alpha=(d_j-t)\cdot m(\hat{T}_t)$, which contradicts the fact that $m(\hat{T}_t)$ machines are enough to accommodate $\hat{T}_t$.
\end{proof}

\begin{proof}[Proof of Lemma~\ref{lemma:partition-at-time-t}]
Recall the construction of the different $S_i$ and consider the state of $S_1,\dots,S_{h_t-1}$ when $S_{h_t}$ is opened. Then there exists some job $j$ which could not be added into any of the subset $S_1$ to $S_{h_t-1}$. As before, we let $\lambda_i$ be the earliest-deadline job from $S_i$. Obviously job $\lambda_i$ is added to $S_i$ before we consider job $j$, $\lambda_i<j$. According to Lemma~\ref{lemma:length-shrink}, there are at least $h_t-1-2\cdot m(\hat{T}_t)/\alpha$ jobs among $\lambda_1,\dots,\lambda_{h_t-1}$ whose feasible time window has a length at least $2(d_j-t)$, and has a remaining laxity no more than $d_j-t$. Thus if we consider the interval $[t,t+2(d_j-t)]$, each of the $h_t-1-2\cdot m(\hat{T}_t)/\alpha$ contribute at least a processing time of $d_j-t$, implying that the workload that has to be finished in this interval is at least $(h_{t}-1-2\cdot m(\hat{T}_t)/\alpha)(d_j-t)$. On the other hand, as $m(\hat{T}_t)$ machines are enough to accommodate all the jobs, the workload processed during $[t,t+2(d_j-t)]$ is at most $m(\hat{T}_t)\cdot 2(d_j-t)$. Hence $(h_{t}-1-2\cdot m(\hat{T}_t)/\alpha)(d_j-t)\le m(\hat{T}_t)\cdot 2(d_j-t)$, which implies the lemma.
\end{proof}

We are ready to prove the main theorem of this section:

\begin{proof}[Proof of Theorem~\ref{thm:logn}.]
Recall that the number of machines used by our algorithm is $\OO(m(L)+h_t\cdot \mu_t)$. Using Lemma~\ref{lemma:length-preempted}, it only remains to bound $m(L)$ and $h_t$. We set $\beta=2-\pi^2/6\in(0,1)$. 

We first show $h_{t}\le 1+(2+2/\alpha)\cdot 4m/\beta=\OO(m)$ for all $t$. According to Lemma~\ref{lemma:partition-at-time-t}, it suffices to prove that $m(\hat{T}_t)\le 4m/\beta$. By Lemma~\ref{lemma:length-preempted}, we know that for any job $j$, we have $\ell_j(t)\geq\beta\ell_j$. Consider the instance $J_{\beta}$ as defined in Subsection~\ref{subsec: oa-pow-lax}. Any feasible schedule of $J_{\beta}$ implies a feasible schedule of $\hat{T}_t$, i.e., we get $m(\hat{T}_t)\le m(J_{\beta})$. Theorem~\ref{thm:laxity-drop} implies that $m(J_{\beta})\le 4m/\beta$ and the claim follows.

Moreover, we show $m(L)\le 4m/\beta=\OO(m)$. Consider any job $j_t\in L$. If job $j$ is released at time $t$, then $\ell_j(t)=\ell_j$. 
Otherwise  we have $\ell_j(t-1)\geq\beta\ell_j$ according to Lemma~\ref{lemma:length-preempted}. As $j$ is $\alpha$-tight at $t-1$ and becomes $\alpha$-loose at $t$, it follows that $j$ is processed in $[t-1,t]$, implying that $\ell_j(t)=\ell_j(t-1)\geq\beta\ell_j$. Thus it holds that $\ell_j(t)\geq\beta\ell_j$. As a consequence, any feasible schedule of $J_{\beta}$ implies a feasible schedule of $L$, implying that $m(L)\le m(J_{\beta})\le 4m/\beta$ by Theorem~\ref{thm:laxity-drop}.
\end{proof}



\section{Lower Bounds}
\label{sec:LB}

\subsection{Lower bound for \LLF}

Phillips et al.~\cite{phillipsSTW02} showed that the competitive ratio of \LLF for (semi-)online machine minimization is not constant. We extend their results by showing that \LLF requires $\Omega({n^{1/3}})$ machines.

\begin{theorem}\label{thm:LLF-lower-bound}
  There exists an instance of $n$ jobs with a feasible schedule on $m$ machines for which \LLF requires $\Omega({n^{1/3}})$ machines.
\end{theorem}
\begin{proof}
Let $m$ be even. For any $\hat{c}>1$, we give an instance at which \LLF using $\hat{c}m=cm/2$ machines fails. Towards this, consider the integer sequence $(x_0,x_1,x_2,\dots)$ with $x_0$ large enough and $x_r=x_0\sum_{i=r+2}^{\infty}(1/c)^i$, for all $r\ge 1$. Let $G_t(t,r)$ be the set of $m/2$ identical {\em (more tight)} jobs with feasible time window $[t,t+x_0\sum_{i=r}^{\infty}(1/c)^i]$ and laxity $x_0\sum_{i=r+1}^{\infty}(1/c)^i$. Moreover, let $G_l(t,x_r)$ be the set of $cm/2$ identical {\em (more loose)} jobs with feasbile time window $[t,t+cx_r]$ and processing time $x_r$. We construct the instance in $k$ rounds (where $k$ is yet to be specified) as follows.

\begin{itemize}
\item {\em Round 1.} From time $0$ to time $x_0$, we release $G_t(0,0)$ and $G_l(t,x_1)$ for $t=icx_1$ where $i=0,1,\cdots,x_0/(cx_1)-1$. It can be easily verified that \LLF will always preempt jobs in $G_t(0,0)$ in favor of jobs in $G_l(t,x_1)$. Thus at time $x_0$, each job in $G_t(0,0)$ is preempted for exactly $x_1\cdot x_0/(cx_1)=x_0/c$ time units, i.e., by time $x_0$ there are still $m/2$ jobs in $G_t(0,x_0)$, each having a remaining processing time of $x_0/c$ to be scheduled within $[x_0,x_0+x_0\sum_{i=1}^{\infty}(1/c)^i]$. However, the optimum solution could finish all the jobs released so far by time $x_0$. 

\item {\em  Round $r>1$.} Carry on the above procedure. Suppose at time $x_0\sum_{i=0}^{r-2}(1/c)^i$ the optimum solution could finish all the jobs released so far while for \LLF there are still $(r-1)\cdot m/2$ jobs, each having a remaining processing time $x_0/c^{r-1}$ to be scheduled within $[x_0\sum_{i=0}^{r-2}(1/c)^i,x_0\sum_{i=0}^{\infty}(1/c)^i]$. Then from time $x_0\sum_{i=0}^{r-2}(1/c)^i$ to time $x_0\sum_{i=0}^{r-1}(1/c)^i$, we release $G_t(x_0\sum_{i=0}^{r-2}(1/c)^i,r-1)$ and $G_l(t,x_{r})$ for $t=x_0\sum_{i=0}^{r-2}(1/c)^i+j\cdot cx_r$ where $j=0,1,\cdots,x_0/(c^{r}x_{r})$.
It can be easily seen that at time $x_0\sum_{i=0}^{r-2}(1/c)^i$ there are in total $(r)\cdot m/2$ jobs having a processing time $x_0/c^{r-1}$ to be scheduled within $[x_0\sum_{i=0}^{r-2}(1/c)^i,x_0\sum_{i=0}^{\infty}(1/c)^i]$. Each of these jobs will be preempted in favor of jobs in $G_l(t,x_{r})$. Thus until time $x_0\sum_{i=0}^{r-2}(1/c)^i$, it is preempted for $x_{r}\cdot x_0/(c^{r}x_{r})=x_0/c^{r}$ time units. This implies that there are $rm/2$ jobs, each having a processing time $x_0/c^r$ to be scheduled within $[x_0\sum_{i=0}^{r-1}(1/c)^i,x_0\sum_{i=0}^{\infty}(1/c)^i]$. 
However, the optimum solution could finish all the jobs released so far by time $x_r$.
\end{itemize}

We estimate the number of jobs released at each round by $m+cm\cdot x_0/(c^{k+2}x_{k+2})=c(c-1)m+m=\OO(c^2m)=\OO(\hat{c}^2m)$. Thus, until round $k$ we release in total $\OO(k\hat{c}^2m)$ jobs. It is easy to see that \LLF requires $\OO(km)$ machines for the remaining jobs and $\hat{c}m$ machines will be insufficient if we choose some $k=\OO(\hat{c})$. Thus it suffices to choose some $n=\OO(\hat{c}^3)$ for the lower bound, i.e., \LLF requires $\Omega(cm)=\Omega(n^{1/3})$ many machines.
\end{proof}

\subsection{Lower bound for Deadline-Ordered Algorithms}

Consider deadline-ordered algorithms that schedule jobs using only the relative order of jobs instead of their actual values. \EDF belongs to this class of algorithms. Lam and To~\cite{lamT99} derived a lower bound on the speed that is necessary to feasibly schedule jobs on~$m$ machines. We modify their instance and argumentation to give the following lower bound on the number of unit-speed machines.

\begin{theorem}\label{thm:deadline-ordered}
  There are instances of~$n$ jobs with a feasible schedule on~$m$ machines for which any deadline-ordered algorithm requires at least~$n-1$ machines.
\end{theorem}

\begin{proof}
We define a collection of instances with identical job sets that all can be feasibly scheduled on~$m$ machines. The only difference in the instances are the deadlines but the relative order is again the same. Thus, a deadline-ordered algorithm cannot distinguish the instances and must produce the same schedule for all of them. We show that this property enforces that (nearly) each job must run on its own machine.

For each~$k\in\{1,2,\ldots,n-m\}$ we define the job set~$J_k$ as follows:
\begin{align*}
  &\textup{For }  1\leq j\leq m: \quad & r_j=0,\ &p_j=1,\ d_j= \left(\frac{m}{m-1}\right)^k  =: \bar{d}_k.\\
  &\textup{For }  m+1 \leq j \leq m+k: \quad &r_j=0,\ &p_j=\left(\frac{m}{m-1}\right)^{j-m},\ d_j= 
\bar{d}_k.\\
  &\textup{For } m+k+1 \leq j \leq n: \quad  &r_j=0,\ &p_j=\left(\frac{m}{m-1}\right)^{j-m},\ d_j= \left(\frac{m}{m-1}\right)^{n-m} =: \bar{d}.
\end{align*}

For every instance~$J_k$, $k\in\{1,2,\ldots,n-m\}$, there is a feasible solution on $m$ machines. Schedule the first $m+k$ jobs in order of their indices one after the other filling a machine up to the deadline~$\bar{d}_k$ before opening the next machine. Since every job has a processing time less than~$\bar{d}_k$, no job will run simultaneously on more than one machine. The total processing volume is 
\begin{equation*}
  \sum_{j=1}^{m+k}p_j = m+\sum_{\ell=1}^k \left(\frac{m}{m-1}\right)^{\ell} 
  = m+ \frac{\left(\frac{m}{m-1}\right)^{k+1}-1}{\left(\frac{m}{m-1}\right)-1} -1
  = m \left(\frac{m}{m-1}\right)^k = m \cdot \bar{d}_k,
\end{equation*}
and thus, all jobs are feasibly scheduled on $m$ machines. The remaining jobs are scheduled by the same procedure in the interval $[\bar{d}_k,\bar{d}]$. With a similar argumentation as above, this is a feasible schedule. In particular, the total processing volume is
\begin{equation*}
  \sum_{j=m+k+1}^{n}p_j = \sum_{\ell=1}^{n-m-k} \left(\frac{m}{m-1}\right)^{\ell+k} 
  = m \left( \left(\frac{m}{m-1}\right)^{n-m} - \left(\frac{m}{m-1}\right)^k \right) = m \cdot (\bar{d}-\bar{d}_k).
\end{equation*}

We now show that any fixed schedule  must use~$n-1$ machines to guarantee a feasible solution for all instances~$J_k$. With the argumentation above this implies the lower bound for any deadline-ordered algorithm. 

W.l.o.g.~we may assume that the order of job indices is the order of deadlines. We could enforce this explicitly by adding a small value, but omit it for the sake of presentation. 

Any deadline-ordered solution must complete the first $m+k$ jobs by time~$\bar{d}_k$ to be feasible for~$J_k$ for~$k\in\{1,\ldots,n-m\}$. Since in any instance~$J_k$ the largest among those jobs has processing time $p_{m+k}=\bar{d}_k$, each job with index $m+k$ must receive its exclusive machine. That means, we need~$n-m$ machines for jobs~$\{m+1,m+2,\ldots,n\}$. The remaining~$m$ jobs with unit processing time must finish in instance~$J_1$ by time~$\bar{d}_1=\frac{m}{m-1}$ which requires~$m-1$ machines. In total, a deadline-ordered algorithm needs~$n-1$ machines whereas optimal solutions with~$m$ machines exist.
\end{proof}

Notice that this very fatal lower bound is achieved already by instances in which all jobs are released at the same time.  Thus, it is not the lack of information about further job arrivals but the lack of precise deadline information that ruins the performance.

\section{Concluding Remarks}

We contribute new algorithmic results for a fundamental online resource minimization problem. We give the first constant competitive algorithms for a major subclass in which jobs have agreeable deadlines, 
and also improve on upper bounds for the general problem. It remains as a major open question if the general preemptive problem admits a constant competitive ratio. We believe that some of our techniques may be useful for the closely related online throughput maximization problem, where we relax the requirement that all jobs must be scheduled (and fix a number of machines). It would also be interesting to quantify the tradeoff between the number of machines and the guaranteed~throughput.

\section*{Acknowledgment} We thank Naveen Garg for discussions on deadline-ordered algorithms and his contribution to the lower bound in Theorem~\ref{thm:deadline-ordered}. Moreover, we thank Benjamin Müller for interesting discussions at an early stage of the project. As part of his Master's thesis he contributed to Thm.~\ref{thm: black box} and part of Thm.~\ref{thm:preempt-samesize-optunknown}.

\bibliographystyle{abbrv}
\bibliography{machinenum}

\end{document}